\documentclass[11pt]{article}
\usepackage[letterpaper,margin=1in]{geometry}

\usepackage{amsmath,amssymb,amsthm,amsfonts,latexsym,bbm,xspace,graphicx,float,mathtools}
\usepackage{braket}
\usepackage{xcolor}
\usepackage{enumitem}

\usepackage{dsfont}

\usepackage{array}

\usepackage[pagebackref,colorlinks,citecolor=blue,linkcolor=black,bookmarks=true,unicode]{hyperref}
\hypersetup{
pdftitle={Single-copy stabilizer testing},
pdfauthor={Marcel Hinsche, Jonas Helsen},
pdfkeywords={quantum information, quantum property testing, stabilizer states, Clifford group}
}
\usepackage[nameinlink,capitalize]{cleveref}
\crefname{equation}{Eq.}{Eqs.}
\Crefname{equation}{Eq.}{Eqs.}

\newtheorem{theorem}{Theorem}[section]
\newtheorem{lemma}[theorem]{Lemma}

\newtheorem{proposition}[theorem]{Proposition}
\newtheorem{fact}[theorem]{Fact}
\newtheorem{corollary}[theorem]{Corollary}

\newtheorem{infthm}[theorem]{Informal Theorem}

\theoremstyle{definition}
\newtheorem{definition}[theorem]{Definition}

\newcommand{\tr}{\mathrm{tr}}
 
\newcommand{\trace}{\tr}
\newcommand{\Ex}{\mathop{\bf E\/}}
\newcommand{\Es}[1]{\mathop{\bf E\/}_{{\substack{#1}}}}
\newcommand{\stab}{\mathrm{Stab}}
\newcommand{\stabn}{\mathrm{Stab}\left(n\right)}
\newcommand{\Cln}{\mathrm{Cl}\left(n\right)}
\newcommand{\Clnk}{\mathrm{Cl}\left(n,k\right)}

\newcommand{\F}{\mathbb{F}_2}
\newcommand{\Ft}{\mathbb{F}_2^t}
\newcommand{\Fn}{\mathbb{F}_2^n}
\newcommand{\Ftwon}{\mathbb{F}_2^n}

\newcommand{\spanP}{\mathbb{P}_{K}}
\newcommand{\avgspanP}{\overline{\mathbb{P}_{K}}}

\renewcommand{\vec}{\mathbf}

\title{Single-copy stabilizer testing}

 \author{Marcel Hinsche\thanks{Freie Universität Berlin  \& 
 IBM Quantum, IBM Research – Zurich, \texttt{m.hinsche@fu-berlin.de}} \and Jonas Helsen\thanks{Centrum Wiskunde \& Informatica (CWI) and QuSoft, Amsterdam, \texttt{jonas@cwi.nl}}}

\begin{document}

\maketitle

\begin{abstract}
\noindent We consider the problem of testing whether an unknown $n$-qubit quantum state $\ket{\psi}$ is a stabilizer state, with only single-copy access. We give an algorithm solving this problem using $O(n)$ copies, and conversely prove that $\Omega(\sqrt{n})$ copies are required for any algorithm.
The main observation behind our algorithm is that when repeatedly measuring in a randomly chosen stabilizer basis, stabilizer states are the most likely among the set of all pure states to exhibit linear dependencies in measurement outcomes. Our algorithm is designed to probe deviations from this extremal behavior.

For the lower bound, we first reduce stabilizer testing to the task of distinguishing random stabilizer states from the maximally mixed state. We then argue that, without loss of generality, it is sufficient to consider measurement strategies that a) lie in the commutant of the tensor action of the Clifford group and b) satisfy a Positive Partial Transpose (PPT) condition. By leveraging these constraints, together with novel results on the partial transposes of the generators of the Clifford commutant, we derive the lower bound on the sample complexity.
\end{abstract}
\newpage
\tableofcontents

\section{Introduction}\label{sec:introduction}
Understanding properties of unknown quantum systems is crucial for applications ranging from benchmarking quantum devices to exploring fundamental physics. One key property of quantum states is their non-stabilizerness which quantifies their deviation from the set of stabilizer states. This property is intricately linked to the ability to enhance Clifford operations, ultimately enabling universal quantum computation. In this work, we are concerned with characterizing this property. Specifically, we study \textit{stabilizer testing}, which has been extensively studied recently within the field of quantum property testing~\cite{montanaroSurveyQuantumProperty2016, grossSchurWeylDualityClifford2021a, buStabilizerTestingMagic2023, grewalImprovedStabilizerEstimation2024a, arunachalamTolerantTestingStabilizer2024b, chenStabilizerBootstrappingRecipe2024}. 

\begin{definition}[Stabilizer testing]
\label{def:stabilizer-testing-informal}
Given access to copies of an unknown $n$-qubit pure state $\ket\psi$, decide if $\ket\psi$ is a stabilizer states or at least $\epsilon$-far from all stabilizer states.
\end{definition}

Motivated by practical feasibility considerations, recent research has made significant progress in understanding the statistical overheads associated with restricted quantum learning and testing protocols \cite{bubeckEntanglementNecessaryOptimal2020a, chenHierarchyReplicaQuantum2021b, huangQuantumAdvantageLearning2022c, loweLowerBoundsLearning2022, chenEfficientPauliChannel2023, oufkirSampleOptimalQuantumProcess2023, harrowApproximateOrthogonalityPermutation2023, caroLearningQuantumProcesses2024a, kingTriplyEfficientShadow2024, liuRoleSharedRandomness2024, chenOptimalTradeoffsEstimating2024a}.
In learning theory and property testing one often considers restricting the learning/testing algorithm to process only one copy of the unknown state at a time. We refer to such algorithms as \textit{single-copy} algorithms. This in contrast to \textit{multi-copy} algorithms that can jointly operate across multiple copies at once. In practice, this limitation may arise due to a lack of quantum memory to store multiple copies which is why in some works single-copy algorithms are referred to as learning algorithms without quantum memory~\cite{chenExponentialSeparationsLearning2022a,huangQuantumAdvantageLearning2022c}
. It may however also arise due to the difficulty of implementing joint operations across multiple copies. The restriction to single-copy algorithms has been shown to lead to large, often even exponential, increases in the number of copies necessary for certain learning or testing tasks when compared to general multi-copy algorithms.

In the context of \emph{learning} stabilizer states, this separation has been tightly characterized (c.f. \Cref{tab:overview-stabilizer-testing-results}). However, for stabilizer testing, an understanding of the separation between single- and multi-copy algorithms is still missing. While Ref.~\cite{grossSchurWeylDualityClifford2021a} gave a $2$-copy algorithm using in total $6$ copies of the unknown state (in the case of qubits), thereby settling the multi-copy complexity, the single-copy complexity remained open.

\subsection{Results}
In this work, we address this gap: First, we give a single-copy stabilizer testing algorithm using $O(n)$ copies, based on the computational difference sampling primitive introduced in \cite{grewalEfficientLearningQuantum2024d}.

\begin{infthm}[Upper bound]\label{infthm:upper}
    There exists a single-copy algorithm for stabilizer testing that uses $t=O(n)$ copies running in time $O(n^3)$.
\end{infthm}

We note that sample complexity of this algorithm beats a brute-force approach using classical shadows \cite{huangPredictingManyProperties2020c}. This approach would consist of estimating the fidelity with all $2^{O(n^2)}$ many $n$-qubit stabilizer states and would require $t=O(n^2)$ many copies (as well as exponential time complexity).\\

Secondly, we provide a lower bound of $\Omega(\sqrt{n})$ copies on the sample complexity of any single-copy algorithm, using the representation theory of the Clifford group \cite{grossSchurWeylDualityClifford2021a} and a Positive Partial Transpose (PPT) relaxation technique due to Harrow \cite{harrowApproximateOrthogonalityPermutation2023}:

\begin{infthm}[Lower bound]
       Any single-copy algorithm for stabilizer testing requires at least $t=\Omega(\sqrt{n})$ copies.
\end{infthm}

As far as we know, this is the first work showing such a fine-grained single-copy/multi-copy separation for a property testing problem, as opposed to the exponential separations usually seen in the literature. We remark that establishing this lower bound for stabilizer testing requires tools well beyond those used to show the lower bounds for stabilizer state learning. The latter bounds follow from the Holevo bound in the multi-copy case \cite{montanaroLearningStabilizerStates2017a} or
require computing only second moments in the single-copy case, see Theorem 5 in \cite{arunachalamOptimalAlgorithmsLearning2023a}.
In contrast, we use $t=\Omega(\sqrt{n})$-th moments of the Clifford group.

\renewcommand{\arraystretch}{1.5}
\begin{table}[h]
\small
\centering
\begin{tabular}{|>{\centering\arraybackslash}m{3cm}|>{\centering\arraybackslash}m{6cm}|>{\centering\arraybackslash}m{3cm}|}
    \hline
    & \textbf{Single-copy} & \textbf{Multi-copy} \\ \hline
    \textbf{Learning} & $\Theta(n^2)$ \cite{aaronsonIdentifyingStabilizerStates2008,arunachalamOptimalAlgorithmsLearning2023a,grewalEfficientLearningQuantum2024d} & $\Theta(n)$ \cite{montanaroLearningStabilizerStates2017a} \\ \hline
    \textbf{Testing} & $\Omega(\sqrt{n})\leq t\leq O(n)$ (\emph{This work}) & $\Theta(1)$ \cite{grossSchurWeylDualityClifford2021a}\\ \hline
\end{tabular}
 \caption{Upper and lower bounds on the number of copies $t$ required for learning and testing stabilizer states with single-copy and multi-copy access, respectively.\label{tab:overview-stabilizer-testing-results}}
\end{table}

\subsection{Technical overview}
\paragraph{Upper bound:}
The single-copy stabilizer testing algorithm we present is very much inspired by the learning algorithms for stabilizer states from \cite{aaronsonIdentifyingStabilizerStates2008,grewalEfficientLearningQuantum2024d,chiaEfficientLearning$t$doped2024}. In particular, in our algorithm, we utilize the \textit{computational difference sampling} primitive introduced by \cite{grewalEfficientLearningQuantum2024d}. It involves two copies of the unknown state $\ket{\psi}$, sampling each in the computational basis, and adding the outcome bitstrings mod 2 during classical post-processing (see \Cref{ssec:computational-difference-sampling} for more details).

Letting $\ket{\psi}$ be the unknown state we want to test, our single-copy algorithm proceeds as follows:
\begin{enumerate}
    \item Sample a uniformly random Clifford $C$ from the Clifford group $\Cln$.
    \item Use $ K  = O(n)$ (but $K>n)$ copies of $\ket{\psi}$ to repeatedly measure $C\ket{\psi}$ via computational difference sampling obtaining samples $\vec v_{1},\dots, \vec v_{K}$.
    \item Taking the samples as rows of a matrix, compute the rank of the $K\times n$ binary matrix in order to check if the samples span $\Fn$ or only a proper subspace of $\Fn$. In the former case, output 1, in the latter case output 0.
\end{enumerate}
We refer to the expectation value (over the choice of random Clifford) of this random process as the \textit{average spanning probability}:
\begin{equation}
    \avgspanP \left(\ket{\psi}\right):= \Es{C\sim \Cln} \left[ \Pr_{\vec v_{1},\dots, \vec v_{K}\sim r_{C \ket \psi}}\mathrm{span}\{ \vec v_{1},\dots,\vec v_{K} \}= \Fn  \right]\,
\end{equation}
Here, $r_{C \ket \psi}$ denotes the computational difference sampling distribution corresponding to the state $C\ket{\psi}$. By repeating the above 3-step process $O(1/\epsilon^2)$ many times, our algorithm estimates $\avgspanP \left(\ket{\psi}\right)$ up to additive precision $\epsilon$.

As we demonstrate, the average spanning probability $\avgspanP \left(\ket{\psi}\right)$ is extremal for stabilizer states, taking a value of roughly $\avgspanP \left(\ket{S}\right) \approx 0.42$ (for any $n\geq 10)$ and increases beyond this value for non-stabilizer states. Our key technical contribution is to show that for a state that is $\epsilon$-far from all stabilizer states, this increase is sufficiently large, namely of the order $\Omega(\epsilon)$. Hence, this increase can be detected sample-efficiently by estimating $\avgspanP(\ket{\psi})$ and comparing to the value for stabilizer states which can be efficiently computed exactly analytically. This gives rise to a single-copy stabilizer testing algorithm.

Understanding $\avgspanP$ for stabilizer states is easy because their output distributions in the computational basis are well-known to be uniform over affine subspaces of $\mathbb{F}_2^n$ (c.f. \Cref{eq:affine-subspace-stabilizer-states}). The key challenge is to analyze and bound the average spanning probability for non-stabilizer states $\ket{\psi}$ because their output distributions in the computational basis are no longer uniform over a subspace. To overcome this challenge, we make use of a key relation between the computational difference sampling distribution $r_\psi$ and the so-called \textit{characteristic distribution} $p_\psi$. The  characteristic distribution $p_\psi$ associated to a pure $n$-qubit state is the distribution over $\mathbb{F}_2^{2n}$ given by
\begin{equation}
p_{\psi}\left(\vec x\right)=\frac{1}{2^{n}}\left| \bra{\psi}P_{\vec x}\ket{\psi}\right|^{2}\,, 
\end{equation}
where we identify Pauli operators in $\{I,X,Y,Z\}^n$ with bitstrings $\vec x = (\vec a, \vec b)$ via $P_{\vec x} =i^{\vec a \cdot \vec b} X^{\vec a}Z^{\vec b}$.
The relation was established in \cite{grewalEfficientLearningQuantum2024d} and posits that the computational difference sampling distribution $r_\psi$ can be written in terms of $p_\psi$ as follows:
\begin{equation}
r_{\psi}\left(\vec a\right)=\sum_{\vec b\in\mathbb{F}_{2}^{n}}p_{\psi}\left(\vec a,\vec b\right)\,.
\end{equation}
Using this relation, we are able to rewrite $\avgspanP \left(\ket{\psi}\right)$ in terms of $p_\psi$. Lastly, to connect back to stabilizer testing, we use the fact that the weight of $p_\psi$ on any commuting subgroup of $\{I,X,Y,Z\}^n$ can be related to the \textit{stabilizer fidelity} of $\ket{\psi}$, defined as $\max_{\ket S\in\mathrm{Stab}\left(n\right)}\left|\braket{S|\psi}\right|^{2}$. In particular, for any such stabilizer group corresponding to an isotropic subspace $M\subset \mathbb{F}_2^{2n}$, it holds that \cite{grossSchurWeylDualityClifford2021a, grewalImprovedStabilizerEstimation2024a},
\begin{equation}
F_{\mathrm{Stab}}\left(\ket \psi\right)\geq \sum_{\vec x\in M} p_{\psi}\left(\vec x \right)\,.
\end{equation}

\paragraph{Lower bound:}
For the lower bound, we first establish a reduction between stabilizer testing and the task of distinguishing uniformly random stabilizer states from the maximally mixed state. This reduction is analogous to the reduction between purity testing and the task of distinguishing Haar random states from the maximally mixed state, which was previously considered in \cite{chenExponentialSeparationsLearning2022a, harrowApproximateOrthogonalityPermutation2023,chenOptimalTradeoffsEstimating2024a}. 
The difference is that the unitary group $U(n)$ is replaced with the Clifford group $\Cln$. Based on this analogy, we adapt the proof strategy due to Harrow \cite{harrowApproximateOrthogonalityPermutation2023} to the Clifford group. We now outline the main ideas behind this proof strategy in more detail.

Consider a single-copy algorithm with access to $t$ copies of either a random stabilizer state $\ket{S}$ or the maximally mixed state.
The first main insight is that since the states $\Es{\ket{S} \sim \stabn}\left[\ket{S}\!\bra{S}^{\otimes{t}} \right]$ and $\mathds{1}^{\otimes t}/2^{nt}$, which the algorithm is supposed to distinguish, both commute with all $C^{\otimes t}$ where $C\in \Cln$, we can restrict our attention to measurements in the commutant of the $t$-fold tensor action of the Clifford group, defined as
\begin{equation}
        \mathrm{Comm}(\Cln, t) := \{A \in \mathcal{L}((\mathbb{C}^{2})^{\otimes n})^{\otimes t}) \;|\; [A, C^{\otimes t}] = 0 \quad \forall C \in \Cln \} \,.
    \end{equation}
This commutant was fully characterized in the seminal work \cite{grossSchurWeylDualityClifford2021a}. In particular, it is spanned by operators $R(T)$ which are associated with certain subspaces $T \subset \mathbb{F}^{2t}_2$. Similar to the permutation operators spanning the commutant of the tensor action $U^{\otimes t}$ of the unitary group, these operators $R(T)$ are also approximately orthogonal with respect to the Hilbert-Schmidt inner product in the regime where $t\ll n$.

The second main insight is that any distinguishing two-outcome POVM $\{M_0, M_1\}$
corresponding to a single-copy measurement strategy is also \emph{Positive Partial Transpose} (PPT), i.e., it satisfies
\begin{equation}\label{eq:ppt-constraint-in-intro}
     0\preceq M_i^{\Gamma_{S}}\preceq I\quad\forall S\subseteq[t]\,.
\end{equation}
Here, $^{\Gamma_{S}}$ denotes taking the partial transpose with respect to the subset of copies indexed by some $S\subset \{1,\dots, t\}$. Motivated by this insight, we study the partial transposes of the unitary operators $R(O)$ which form an important subgroup of the commutant (see Section 4.2 in \cite{grossSchurWeylDualityClifford2021a}). In particular, we are able to characterize the singular values of the partial transposes $R(O)^{\Gamma_{S}}$.

Finally, using our insights about the partial transposes, we are able to leverage both the PPT constraint in \Cref{eq:ppt-constraint-in-intro} and the approximate orthogonality of the operators $R(T)$ to derive the sample complexity lower bound.

\subsection{Related work}
\paragraph{Learning stabilizer states and their generalizations:}
Our work is closely related to a line of work exploring the learnability in a tomographic sense of stabilizer states and their generalizations. This line started with \cite{aaronsonIdentifyingStabilizerStates2008, montanaroLearningStabilizerStates2017a} giving single- and multi-copy algorithms for learning stabilizer states, respectively. Since then, these works were generalized to learning states with large stabilizer dimension \cite{laiLearningQuantumCircuits2022a,grewalEfficientLearningQuantum2024d, chiaEfficientLearning$t$doped2024, leoneLearningTdopedStabilizer2024a}, meaning states stabilized by a non-maximal abelian subgroup of Pauli operators, and to higher degree binary phase states \cite{arunachalamOptimalAlgorithmsLearning2023a}, where degree 2 corresponds to stabilizer states. In all of these works, the underlying assumption is that the unknown state $\ket{\psi}$ belongs to the restricted class of states that is to be learned. This assumption sets these works apart from a testing scenario where the very goal is to test the validity of such an assumption.

\paragraph{Agnostic tomography of stabilizer states:} Recent works have shifted towards agnostic learning of states (also called agnostic tomography) where $\ket{\psi}$ can be arbitrary and the goal is to learn the best approximation from a given class of states. This agnostic learning framework naturally also gives rise to testing algorithms. For instance, \cite{grewalImprovedStabilizerEstimation2024a} addressed agnostic tomography of stabilizer states, \cite{grewalAgnosticTomographyStabilizer2024} focused on stabilizer product states, and \cite{chenStabilizerBootstrappingRecipe2024} extended the approach to stabilizer states, states with large stabilizer dimensions, and product states.

\paragraph{Tolerant stabilizer testing:}
Another exciting direction is \textit{tolerant} stabilizer testing, where the goal is to determine if an unknown state $\ket{\psi}$ is approximately a stabilizer state, rather than exactly one or far from it. Recent work \cite{grewalImprovedStabilizerEstimation2024a, arunachalamTolerantTestingStabilizer2024b}, building on \cite{grossSchurWeylDualityClifford2021a}, has shown that by repeatedly running the algorithm from \cite{grossSchurWeylDualityClifford2021a}, one can achieve tolerant testing across a broad range of parameters.

\paragraph{Measuring non-stabilizerness:}
Characterizing the amount of non-stabilizerness in a quantum state is also an active topic of research in a more physics-oriented literature. In particular, \cite{haugScalableMeasuresMagic2023a, haugEfficientQuantumAlgorithms2024a} showed how 2-copy Bell measurements can be used to quantify non-stabilizerness in quantum states. However, to the best of our knowledge, there is to date no efficient single-copy protocol for measuring non-stabilizerness. While it has been realized that the so-called \textit{stabilizer entropies}~\cite{leoneStabilizerEnyiEntropy2022b} can also be measured via single-copy protocols \cite{leoneStabilizerEnyiEntropy2022b,olivieroMeasuringMagicQuantum2022a}, the proposed protocols feature an exponential sample complexity.

\subsection{Open Problems}\label{ssec:open-problems}
As a result of this work, we identified multiple interesting open problems:
\begin{enumerate}
    \item \textbf{Matching upper and lower bounds:} Currently, our $\Omega(\sqrt{n})$ lower bound and our $O(n)$ upper bound on the sample complexity are not quite matching. On the one hand the lower bound can potentially be tightened by a sharper analysis of the orthogonality of the operators spanning the commutant $\mathrm{Comm}(\Cln, t)$ as well as an even more refined understanding of their partial transposes. On the other hand, we found a surprising $O(1)$ \emph{post-selective} single-copy algorithm for distinguishing stabilizer states from the maximally mixed state (see \Cref{ssec:tree-representation-framework} for more details). This could be indicative of the existence of single-copy algorithms with sub-linear scaling. 
    \item \textbf{Lower bound via the tree representation framework:} As mentioned before, \cite{chenExponentialSeparationsLearning2022a} obtains a lower bound for purity testing by considering an analogous distinguishing task to the one we consider. However, they bound the sample complexity of this distinguishing task via the so-called tree representation framework instead of the PPT relaxation. We discuss the obstacles to this approach in our case in \Cref{ssec:tree-representation-framework}. It remains an open question, if our lower bound can also be obtained by means of the tree representation framework.
    \item \textbf{Tolerant single-copy stabilizer testing, efficiently measuring of non-stabilizerness:} Can our single-copy stabilizer testing algorithm be made tolerant? A related open question is: Are there efficient single-copy algorithms for measuring some measure of non-stabilizerness? 
    \item \textbf{Multi-copy access but restricted memory:} Ref.~\cite{chenOptimalTradeoffsEstimating2024a} considers an intermediate scenario between single- and multi-copy access. Therein, multi-copy access is allowed, however the overall quantum memory is only $k<2n$ qubits, so one cannot operate across two full copies of the unknown state $\ket{\psi}$. In this setting, they establish a phase transition for the sample complexity of purity testing. Due to the structural similarities between purity testing and stabilizer testing, it is interesting to try to generalize their result to stabilizer testing. 
\end{enumerate}

\section{Preliminaries}\label{sec:preliminaries}
We begin by setting some notation. We often denote pure
states by $\psi=\ket{\psi}\bra{\psi}$. For a positive integer $n$, we define $\left[n\right]:=\left\{ 1,\dots,n\right\} $.
For a set of vectors $\left\{\vec v_{1},\dots, \vec v_{K}\right\} $, we denote
their span by $\langle \vec v_{1},\dots,\vec v_{K}  \rangle$. We denote by $\F$
the finite field of $2$ elements and by $\Fn$ the
$n$-dimensional vector space over this field. We denote by $\vec 0_n\in\Fn$ the zero-vector in this vector space. For $p$ a distribution
over a set $S$, we denote drawing a sample $x\in S$ according to
$p$ by $x\sim p$. For $p$ a distribution over $\Fn$
and $H\subseteq\Fn$ a subset, we denote the
weight of $p$ on the set $H$ by $p\left(H\right):=\sum_{\vec x\in H}p\left(\vec x\right)$. \\

We will make extensive use of the following standard fact of probability theory on linear spaces.
\begin{lemma}
\label{lem:subspace-weight-spanning-probability}
Let $p$ be a distribution
over $\Fn$ and let $H$ be a subspace of $\Fn$
and let $p(H)$ be the probability weight of $p$ on $H$.
Then, for $\vec x_{1},\dots, \vec x_{K}\sim p$,
\begin{equation}
\Pr_{\vec x_{1},\dots, \vec x_{K}\sim p}\,\left(\langle \vec x_{1},\dots, \vec x_{K}\rangle\subseteq H\right)=\left(p\left(H\right)\right)^{K}
\end{equation}
\end{lemma}

\subsection{Vector spaces over \texorpdfstring{$\mathbb{F}_{2}$}{F2}}

In this work, we will deal repeatedly with the vector spaces $\mathbb{F}_{2}^{n}$
 and $\mathbb{F}_{2}^{2n}$ and two different inner products on them:
\begin{definition}[Standard inner product]
 For $\vec a, \vec b\in\mathbb{F}_{2}^{n}$, we define their \emph{standard
inner product} as 
\begin{equation}
\vec a\cdot \vec b=a_{1}b_{1}+\dots+a_{n}b_{n}
\end{equation}
where operations are performed over $\mathbb{F}_{2}$. 
\end{definition}

Note that the standard inner product is also defined on $\mathbb{F}_{2}^{2n}$.
For instance, for $\vec x=\left(\vec a, \vec b\right), \vec y=\left(\vec v, \vec w\right)\in\mathbb{F}_{2}^{2n}$
we have $\vec x\cdot \vec y= \vec a\cdot \vec v+ \vec b\cdot \vec w$. 
\begin{definition}[Orthogonal complement]
 Let $H\subseteq\mathbb{F}_{2}^{n}$ be a subspace. The \emph{orthogonal
complement }of $H$, denoted by $H^{\perp}$, is defined by
\begin{equation}
H^{\perp}:=\left\{ a\in\mathbb{F}_{2}^{n}: \vec a\cdot \vec b=0\,, \; \forall\,\vec b\in H\,\right\} .
\end{equation}
\end{definition}

\begin{fact}
Let $H$ be a subspace of $\mathbb{F}_{2}^{n}$. Then:
\end{fact}

\begin{itemize}
\item $H^{\perp}$ is a subspace.
\item $\left(H^{\perp}\right)^{\perp}=H$.
\item $\dim\left(H\right)+\dim\left(H^{\perp}\right)=n$
\end{itemize}
For elements of $\mathbb{F}_{2}^{2n}$, we additionally also introduce
the symplectic inner product.
\begin{definition}[Symplectic inner product]
 For $\vec x,\vec y\in\mathbb{F}_{2}^{2n}$, we define their \emph{symplectic
inner product }as 
\begin{equation}
\left[\vec x, \vec y\right]=x_{1}y_{n+1}+x_{2}y_{n+2}+\dots+x_{2n}y_{n}
\end{equation}
where operations are performed over $\mathbb{F}_{2}$.
\end{definition}

A subspace $T\subset\mathbb{F}_{2}^{2n}$ is said to be \emph{isotropic} when for all $\vec x,\vec y\in T$, $\left[\vec x, \vec y\right]=0$. A \emph{Lagrangian
}subspace $M\subset\mathbb{F}_{2}^{2n}$ is an isotropic subspace
of maximal dimension, namely of dimension $\dim\left(M\right)=n$.

\subsection{Phaseless Pauli operators and the characteristic distribution}
In this section we recall some well-known facts about the Pauli group.
The single-qubit Pauli matrices are denoted by $\left\{ I,X,Y,Z\right\} $.
The $n$-qubit Pauli group $\mathcal{P}_{n}$ is the set $\left\{ \pm1,\pm i\right\} \times\left\{ I,X,Y,Z\right\} ^{\otimes n}$.
The Clifford group is the normalizer of the Pauli group. We denote
the $n$-qubit Clifford group by $\Cln$.

In this work, we will deal almost exclusively with phaseless Pauli
operators (also often referred to as Weyl operators, c.f. \cite{grossSchurWeylDualityClifford2021a, grewalEfficientLearningQuantum2024d}),
i.e., elements of the set $\left\{ I,X,Y,Z\right\} ^{\otimes n}$.
Hence, whenever we refer to a Pauli operator, we refer to its phaseless
version unless explicitly stated otherwise. Since $\left\{ I,X,Y,Z\right\} ^{\otimes n}$
is one-to-one with $\mathbb{F}_{2}^{2n}$, we can label elements of
this set by bitstrings of length $2n$ as follows. Let $\vec x=\left(\vec a, \vec b\right)=\left(a_{1},a_{2},\dots,a_{n},b_{1},b_{2}\dots,b_{n}\right)\in\mathbb{F}_{2}^{2n}$.
We then define
\begin{equation}
P_{\vec x}=i^{\vec a\cdot \vec b}\left(X^{a_{1}}Z^{b_{1}}\right)\otimes\cdots\otimes\left(X^{a_{n}}Z^{b_{n}}\right)=i^{\vec a \cdot \vec b} X^{\vec a}Z^{\vec b}\,.
\end{equation}
Here, as an exception, the inner product $\vec a\cdot \vec b$ on the phase in front is understood as being an integer rather than mod 2.
Throughout this work, we will often identify bitstrings in $\vec x\in\mathbb{F}_{2}^{2n}$
with their corresponding (phaseless) Pauli operator $P_{\vec x}$. For
instance, we will identify the set $\mathcal{X}:=\left\{ I,X\right\} ^{\otimes n}$
of Pauli strings made up only of $I$ and $X$ factors with $\Fn\times \vec 0_n$.
Similarly, we will identify $\mathcal{Z}:=\left\{ I,Z\right\} ^{\otimes n}$
with $\vec 0_n \times\Fn$. We also introduce the notation
$\mathcal{Z}^{\times}=\mathcal{Z}\setminus\left\{ I^{\otimes n}\right\} $
which we identify with $\vec 0_n \times\mathbb{F}_{2}^{n}\setminus\left\{ \vec 0_{2n}\right\} $.

 The commutation relations between Pauli operators are captured by
 the symplectic inner product of their corresponding bitstrings. In
 particular, for $\vec x,\vec y\in\mathbb{F}_{2}^{2n}$, the corresponding Pauli
 operators $P_{\vec x},P_{\vec y}$ commute if $\left[\vec x,\vec y\right]=0$ and anticommute
 if $\left[\vec x,\vec y\right]=1$.

Any $n$-qubit state $\rho$ can be expanded in the Pauli basis as
\begin{equation}
\rho=\frac{1}{2^{n}}\sum_{\vec x\in \Ftwon}\trace\left(\rho P_{\vec x}\right)P_{\vec x}\,.
\end{equation}
The squared coefficients of this expansion form a probability distribution
known as the \textit{characteristic distribution}:
\begin{definition}[Characteristic distribution \cite{grossSchurWeylDualityClifford2021a, grewalEfficientLearningQuantum2024d}]
Let $\ket \psi$ be an $n$-qubit pure state, then its characteristic
distribution $p_{\psi}$ is defined via
\begin{equation}
p_{\psi}\left(\vec x\right)=\frac{1}{2^{n}}\trace\left(\psi P_{\vec x}\right)^{2}\,.
\end{equation}
\end{definition}

\noindent Note that this distribution is automatically normalized
because $\frac{1}{2^{n}}\sum_{\vec x}\trace\left(\psi P_{\vec x}\right)^{2}=\trace\left(\psi^{2}\right)=1$.

\subsection{Computational difference sampling}\label{ssec:computational-difference-sampling}

The core measurement routine of the testing algorithm introduced in this paper (Theorem \ref{infthm:upper}) relies on \emph{computational difference sampling}. This primitive
was introduced in \cite{grewalEfficientLearningQuantum2024d}
and refers to measuring two copies of a state in the
computational basis and adding the obtained outcomes
via bit-wise addition mod 2 in classical post-processing. More formally, we have the following definition:
\begin{definition}[Computational difference sampling]
 Computational difference sampling a quantum state $\ket{\psi}$
corresponds to the following quantum measurement. (1) Measure $\ket{\psi}$
in the computational basis to get outcome $\vec a\in\mathbb{F}_{2}^{n}$
. (2) Measure an additional copy of $\ket{\psi}$ in the computational
basis to get another outcome $\vec b\in\mathbb{F}_{2}^{n}$. (3) Output
$\vec a+\vec b\in\mathbb{F}_{2}^{n}$ (addition mod 2). Let $r_{\psi}\left(\vec a\right)$
denote the probability of sampling $\vec a\in\mathbb{F}_{2}^{n}$ through
this process.
\end{definition}

\noindent Importantly, computational difference sampling only uses
single-copy measurements. In this work, we will repeatedly make use
of the following key correspondence between subspace weights of the
computational difference sampling distribution $r_{\psi}$ and the
characteristic distribution $p_{\psi}$:
\begin{lemma}[Subspace weight correspondence between $r_{\psi}$ and $p_{\psi}$]
\label{lem:relation-r-p}
Let $\ket \psi$ be a pure $n$-qubit quantum
state. Given a subspace $H\subseteq\mathbb{F}_{2}^{n}$, consider
its orthogonal complement $H^{\perp}$ (with respect to the standard
inner product on $\mathbb{F}_{2}^{n}$). Then,
\begin{equation}
r_{\psi}\left(H\right)=\left|H\right|\,p_{\psi}\left(\vec 0 _n\times H^{\perp}\right)\,,
\end{equation}
where $\vec 0_n\times H^{\perp}\subseteq\mathcal{Z}$.
\end{lemma}

Importantly, \Cref{lem:relation-r-p} says that the weight of
$r_{\psi}$ on any subspace only depends on the weight of $p_{\psi}$
on a corresponding subspace of Paulis in $\mathcal{Z}$. We note that \Cref{lem:relation-r-p} is already implicit in Corollary 8.5. in \cite{grewalEfficientLearningQuantum2024d}.
However, their statement is presented with respect to the symplectic
complement since they work throughout with the symplectic inner product
on $\mathbb{F}_{2}^{2n}$ instead of the standard inner product on
$\mathbb{F}_{2}^{n}$. Here, we find it more convenient to state this
relationship in terms of the orthogonal complement. The proof of \Cref{lem:relation-r-p} is presented in \Cref{subsec:Proof-of-relation-lemma}
for the sake of completeness.

\subsection{Stabilizer states and isotropic subspaces}

A pure $n$-qubit state is a \emph{stabilizer state} if there exists an Abelian
group $S\subset\mathcal{P}_{n}$ of $2^{n}$ Pauli operators $P\in\mathcal{P}_{n}$
(with phase of $+1$ or $-1$) such that
\begin{equation}
S=\left\{ P\in\mathcal{P}_{n}:P\ket{\psi}=\ket{\psi}\right\}.
\end{equation}
This abelian group is the \emph{stabilizer group} of the stabilizer
state and determines it uniquely. Henceforth, we denote stabilizer
states by $\ket S$ and denote the set of all pure $n$-qubit stabilizer states by $\mathrm{Stab}\left(n\right)$. By considering the phaseless (or unsigned) versions
of the Pauli operators forming a stabilizer group, every stabilizer
group can be associated to a Lagrangian subspace $M$. That is, Lagrangian
subspaces are in a one-to-one correspondence with unsigned stabilizer
groups. There are $2^{n}$ many different stabilizer states
corresponding to the same Lagrangian subspace $M$, with each stabilizer
state corresponding to a different choice of signs of the Pauli operators making up the unsigned stabilizer group.
For more background on the connection between symplectic geometry
and the stabilizer formalism, see \cite{grossHudsonTheoremFinitedimensional2006a,grossSchurWeylDualityClifford2021a}.

Every $n$-qubit pure stabilizer state can (up to a global phase)
be written in the form

\begin{equation}
\label{eq:affine-subspace-stabilizer-states}
\frac{1}{\left|A\right|}\sum_{x\in A}\left(-1\right)^{q\left(x\right)}i^{l\left(x\right)}\ket x \,,
\end{equation}
where $A$ is an affine subspace of $\mathbb{F}_{2}^{2n}$ and $l,q:\mathbb{F}_{2}^{2n}\to\mathbb{F}_{2}$
are linear and quadratic (respectively) polynomials over $\mathbb{F}_{2}$
\cite{dehaeneCliffordGroupStabilizer2003a,vandennestClassicalSimulationQuantum2009}.
This implies that when measured in the computational basis, the output
distributions of stabilizer states are uniform over affine subspaces.
Furthermore, computational difference sampling is convenient precisely because it removes any
affine shift and hence leads to distributions that are uniform over
linear subspaces of $\mathbb{F}_{2}^{n}$.

\subsection{Stabilizer fidelity and stabilizer testing}

We denote the set of all pure $n$-qubit stabilizer states by $\mathrm{Stab}\left(n\right)$.
A rich topic in quantum information is concerned with determining
the nonstabilizerness of quantum states, also commonly referred to
as magic. For more background on a range of magic measures, see e.g. \cite{liuManyBodyQuantumMagic2022b}.

Here, we will specifically focus on the so-called stabilizer fidelity:
\begin{definition}[Stabilizer fidelity, \cite{bravyiSimulationQuantumCircuits2019a}]
Let $\ket\psi$ be a pure $n$-qubit quantum state. The \emph{stabilizer
fidelity} of $\ket\psi$ is defined as 
\begin{equation}
F_{\mathrm{Stab}}\left(\ket\psi\right):=\max_{\ket S\in\mathrm{Stab}\left(n\right)}\left|\braket{S|\psi}\right|^{2}
\end{equation}
\end{definition}

\noindent The stabilizer (in-)fidelity serves as a magic measure for pure states. Arguably, a simpler
problem than trying to determine or estimate the non-stabilizerness
of an unknown quantum state $\ket\psi$, is simply to decide if $\ket\psi$
is a stabilizer state or has at least a bit of magic without necessarily
quantifying how much magic. This task is known as \emph{stabilizer
testing}, and it has received attention in the past in the context
of quantum property testing \cite{grossSchurWeylDualityClifford2021a, buStabilizerTestingMagic2023,grewalEfficientLearningQuantum2024d, arunachalamTolerantTestingStabilizer2024b}.
Formally, the task is the following:
\begin{definition}[Stabilizer testing]
\label{def:stabilizer-testing}
Given $\epsilon>0$ and access
to copies of an unknown $n$-qubit pure state $\ket\psi$, decide if $\ket\psi$
is 
\begin{enumerate}
\item either a stabilizer state, i.e., $F_{\mathrm{Stab}}\left(\ket \psi\right)=1$,
\item or at least $\epsilon$ far from any stabilizer state, i.e., $F_{\mathrm{Stab}}\left(\ket 
 \psi\right)\leq1-\epsilon$.
\end{enumerate}
It is promised that $\ket \psi$ satisfies one of the two cases.
\end{definition}

\noindent In this work, we propose and analyze a single-copy algorithm
for stabilizer testing. Our analysis relies on the fact that the stabilizer
fidelity can be related to the characteristic
distribution.
To state this relation, we start with the following fact:
\begin{fact}[\cite{grossSchurWeylDualityClifford2021a,grewalEfficientLearningQuantum2024d}]
\label{fact:m-psi-isotropic} Let $\ket \psi$ be a pure $n$-qubit quantum
state. Then, consider the set
$M_{\psi}:=\{ \vec y\in\mathbb{F}_{2}^{2n}:\trace\left(\psi P_{\vec y}\right)^{2}>\tfrac{1}{2}\}$
corresponding to the state $\ket\psi$. Then, $M_{\psi}$ is isotropic.
That is, for all $\vec x, \vec y\in M_{\psi}$, $\left[\vec x, \vec y\right]=0$. Furthermore,
$\left|M_{\psi}\right|\leq2^{n}$.
\end{fact}

\noindent Note that, in general, the set $M_{\psi}$ is not necessarily a
subspace, in particular, it need not be closed under addition. However, for a stabilizer
state $\ket S$, its corresponding set $M_{S}$ corresponds precisely
to the Lagrangian subspace associated with the stabilizer state. We
now state a lower bound on the stabilizer fidelity $F_{\mathrm{Stab}}\left(\ket \psi\right)$
of a state $\ket \psi$ in terms of its characteristic distribution $p_{\psi}$
and the corresponding set $M_{\psi}$:

\begin{fact}[Proof of Theorem 3.3 in \cite{grossSchurWeylDualityClifford2021a}, Corollary 7.4. in \cite{grewalImprovedStabilizerEstimation2024a}]
\label{fact:lower-bound-stabilizer-fidelity}Let $\ket \psi$ be an $n$-qubit
pure state and $M$ be a Lagrangian subspace of $\mathbb{F}_{2}^{2n}$,
then
\begin{equation}
F_{\mathrm{Stab}}\left(\ket \psi\right)\geq p_{\psi}\left(M\right)\,.
\end{equation}
In particular, it follows that
\begin{equation}
F_{\mathrm{Stab}}\left(\ket \psi\right)\geq p_{\psi}\left(M_{\psi}\right)\,,
\end{equation}
since $M_{\psi}$ can always be completed to a Lagrangian subspace.
\end{fact}

\subsection{The action of (random) Cliffords on \texorpdfstring{$\mathbb{F}_{2}^{2n}$}{F22n}}

Clifford circuits map Pauli operators to Pauli operators under conjugation.
In particular, when applying a Clifford circuit $C$ to a phaseless
Pauli operator $P_{\vec x}$, the resulting Pauli is not necessarily a
phaseless one. However, when identifying Pauli operators with their
phaseless versions, we can define the Clifford action on $\mathbb{F}_{2}^{2n}$:
We have that $C\left(\vec x\right)=\vec y$ if
\begin{equation}
CP_{\vec x}C^{\dagger}=\pm P_{\vec y}\,.
\end{equation}
Similarly, for a subset $S\subseteq\mathbb{F}_{2}^{2n}$, we write
$C\left(S\right)$ to denote $C\left(S\right)=\left\{ C\left(\vec x\right): \vec x\in S\right\} $.
In fact, the action of the Clifford group on $\mathbb{F}_{2}^{2n}$
is precisely that of the symplectic group $\mathrm{Sp}\left(2n,\mathbb{F}_{2}\right)$.
Hence, in particular, the Clifford action preserves the symplectic
inner product, $\left[C\left(\vec x\right),C\left(\vec y\right)\right]=\left[\vec x,\vec y\right]$.

When we apply a Clifford unitary to a state $\ket \psi$, we shuffle the
individual probabilities of the characteristic distribution $p_{\psi}$
around. In particular, the mapping is precisely described by the Clifford
action on $\mathbb{F}_{2}^{2n}$ since the characteristic distribution
$p_{\psi}$ only depends on the phaseless Pauli operators. Formally,
we have:
\begin{fact}
\label{fact:Clifford-mapping-dagger}Let $\ket{\psi}$ be an $n$-qubit
pure state, let $C\in\Cln$ be a Clifford unitary, then
for all $\vec x\in\mathbb{F}_{2}^{2n}$
\begin{equation}
p_{C\ket{\psi}}\left(\vec x\right)=p_{\psi}\left(C^{\dagger}\left(\vec x\right)\right)\,.
\end{equation}
\end{fact}

\paragraph{Random Cliffords:}

In this paper, we are mainly interested in applying uniformly random
Clifford unitaries to pure states. In particular, a core part of the
analysis of our algorithm relies on the action of random Cliffords,
and hence random symplectic transformations, on isotropic or even
Lagrangian subspaces of $\mathbb{F}_{2}^{2n}$. As demonstrated in
several works \cite{koenigHowEfficientlySelect2014,bravyiHadamardfreeCircuitsExpose2021,bergSimpleMethodSampling2021b},
random Clifford unitaries can be sampled and compiled into Clifford
circuits efficiently:
\begin{fact}
There is a classical algorithm that samples a uniformly random element
$C$ of the $n$-qubit Clifford group $\Cln$ and outputs
a Clifford circuit implementation of $C$ in time $O\left(n^{2}\right)$.
\end{fact}

A uniformly random Clifford $C$ corresponds to a uniformly random
symplectic transformation. When such a transformation acts on an isotropic
subspace $T$ of $\mathbb{F}_{2}^{2n}$, it results in a uniformly
random isotropic subspace of the same dimension. In particular, a
random Clifford $C$ acting on a Lagrangian subspace $M$ results
in a uniformly random Lagrangian subspace $C\left(M\right)$. In this
work, the Lagrangian subspace $\mathcal{Z}=\vec 0_{n}\times\mathbb{F}_{2}^{n}$
plays a distinguished role because of its relation to computational
difference sampling. In particular, let us now answer the following
question: What is the probability that a uniformly random Lagrangian
subspace $C\left(M\right)$ has a $k$-dimensional intersection with
the fixed Lagrangian subspace $\mathcal{Z}$? This is captured by
the following definition:
\begin{definition}[Probability of $k$-dimensional intersection]
\label{def:p-n-k-def}Let $M$ be a Lagrangian subspace of $\mathbb{F}_{2}^{2n}$
and let $0\leq k\leq n$. Then, we define
\begin{equation}
Q\left(n,k\right):=\Pr_{C\sim\Cln}\left(\dim\left(C\left(M\right)\cap\mathcal{Z}\right)=k\right).
\end{equation}
This probability can be expressed as follows,
\begin{equation}
Q\left(n,k\right)=\frac{\left|\mathrm{Cl}\left(n,k\right)\right|}{\left|\mathrm{Cl}\left(n\right)\right|}\,,
\end{equation}
where we introduced 
\begin{equation}
\mathrm{Cl}\left(n,k\right):=\left\{ C\in\mathrm{Cl}\left(n\right):\dim\left(C\left(M\right)\cap\mathcal{Z}\right)=k\right\} \,.
\end{equation}
\end{definition}

To characterize $Q\left(n,k\right)$, we can count the relevant sets
of Lagrangian subspaces. First, we start with the total number of
them:
\begin{fact}
The total number of Lagrangian subspaces of $\mathbb{F}_{2}^{2n}$
is given by
\begin{equation}
\mathcal{T}\left(n\right)=\prod_{i=1}^{n}\left(2^{i}+1\right)=2^{\frac{1}{2}n\left(n+1\right)}\prod_{i=1}^{n}\left(1+\frac{1}{2^{i}}\right).
\end{equation}
\end{fact}

Next, the number of Lagrangian subspaces whose intersection with $\mathcal{Z}$
is $k$-dimensional was obtained already in Corollary 2 in Ref.
\cite{kuengQubitStabilizerStates2015a}:
\begin{lemma}[Corollary 2 in \cite{kuengQubitStabilizerStates2015a}]
 Let $M$ be a fixed Lagrangian subspace of $\mathbb{F}_{2}^{2n}$.
The number of Lagrangian subspaces $N$ whose intersection with $M$
is $k$-dimensional is given by
\begin{equation}
\kappa\left(n,k\right):=\binom{n}{k}_{2}\cdot 2^{\frac{1}{2}\left(n-k\right)\left(n-k+1\right)}\, ,
\end{equation}
where $\binom{n}{k}_{2}$ is the Gaussian binomial coefficient given
by (for $k\leq n$)
\begin{equation}
\binom{n}{k}_{2}=\prod_{i=0}^{k-1}\frac{2^{n-i}-1}{2^{k-i}-1} \,.
\end{equation}
In particular, for $k=0,1,2$, we have
\begin{align}
\kappa\left(n,0\right) & =2^{\frac{1} {2}n\left(n+1\right)} \, ,\\
\kappa\left(n,1\right) & =\left(2^{n}-1\right)\,2^{\frac{1}{2}\left(n-1\right)n} \, ,\\
\kappa\left(n,2\right) & =\frac{\left(2^{n}-1\right)\left(2^{n-1}-1\right)}{3}2^{\frac{1}{2}\left(n-1\right)\left(n-2\right)} \, .
\end{align}
\end{lemma}

Hence, we can characterize $Q\left(n,k\right)$ as follows:
\begin{corollary}
\label{cor:p-n-k-values}Let $0\leq k\leq n$. Then,
\begin{equation}
Q\left(n,k\right)=\frac{\kappa\left(n,k\right)}{\mathcal{T}\left(n\right)} \, .
\end{equation}
In particular, for $k=0,1,2$, we have
\begin{align}
Q\left(n,0\right) & =\prod_{i=1}^{n}\left(\frac{1}{1+2^{-i}}\right)\, ,\\
Q\left(n,1\right) & =\frac{2^{n}-1}{2^{n}}\:Q\left(n,0\right)\, ,\\
Q\left(n,2\right) & =\frac{1}{3}\frac{\left(2^{n}-1\right)\left(2^{n}-2\right)}{2^{2n}}\:Q\left(n,0\right) \, .
\end{align}
\end{corollary}

Note that for fixed $n$, $Q\left(n,k\right)$ forms a probability
distribution over $k\in\left\{ 0,\dots,n\right\} $ and is normalized as
$\sum_{k=0}^{n}Q\left(n,k\right)=1$. 
We remark two important aspects about the distribution $Q\left(n,k\right)$:
First, the distribution converges quickly for increasing $n$, so
that the values of $Q\left(n,k\right)$ do not change significantly
above $n=10$. Secondly, the distribution decays very quickly and,
regardless of the value of $n$, the primary contributions come from
$k=0,1,2,3$ .
The largest contribution is always $Q\left(n,0\right)$ whose limiting
value is given by
\begin{equation}
Q\left(n,0\right)\overset{n\to\infty}{\to}0.41942244...
\end{equation}
It approaches this value monotonously from above. For large enough $n$,
say $n\geq7$, it is useful to keep the following approximate values
in mind $Q\left(n,0\right)\approx0.42$ and $Q\left(n,1\right)\approx Q\left(n,0\right)$
and $Q\left(n,2\right)\approx\frac{Q\left(n,1\right)}{3}$.

\subsection{Commutant of Clifford tensor powers}\label{ssec:commutant}
A key ingredient for deriving our lower bound  is the commutant of the $t$-fold tensor power action of the Clifford group $\Cln$, i.e., the linear space of operators on $(\mathbb{C}^{2})^{\otimes n})^{\otimes t}$ that commute with $C^{\otimes t}$ for all $C \in \Cln$. Formally, we define it as follows:
\begin{definition}[Commutant of $t$-th Clifford tensor power action]
    We define $\mathrm{Comm}(\Cln, t)$ as follows
    \begin{equation}
        \mathrm{Comm}(\Cln, t) := \{A \in \mathcal{L}((\mathbb{C}^{2})^{\otimes n})^{\otimes t}) \;|\; [A, C^{\otimes t}] = 0 \quad \forall C \in \Cln \}.
    \end{equation}
\end{definition}

The seminal work \cite{grossSchurWeylDualityClifford2021a} fully characterized this commutant in terms of so-called stochastic Lagrangian subspaces defined as follows:

\begin{definition}[Stochastic Lagrangian subspaces]
The set $\Sigma_{t,t}$ denotes the set of all subspaces $T\subseteq \mathbb{F}^{2t}_2$ with the following properties:
\begin{enumerate}
    \item $\vec x \cdot \vec x = \vec y \cdot \vec y \mod 4 $ for all $(\vec x, \vec y) \in T$,
    \item $\dim (T) = t$,
    \item $\vec 1_t = (1,\dots,1) \in T$.
\end{enumerate}
We refer to elements in $\Sigma_{t,t}$ as \textit{stochastic Lagrangian subspaces}.
    
\end{definition}

In particular, the key result of \cite{grossSchurWeylDualityClifford2021a} is that the commutant $\mathrm{Comm}(\Cln, t)$ is spanned by operators $R(T)$ associated with the stochastic Lagrangian subspaces $T\in \Sigma_{t,t}$.

\begin{theorem}[Theorem 4.3 in \cite{grossSchurWeylDualityClifford2021a}]
If $n\geq t-1$, then $\mathrm{Comm}(\Cln, t)$ is spanned by the linearly independent operators $R(T) := r(T)^{\otimes n}$, where $T\in\Sigma_{t,t}$ and 
\begin{equation}
    r(T) := \sum_{(\vec x, \vec y) \in T} \ket{x}\bra{y}\,.
\end{equation}
\end{theorem}

Alongside this central characterization, \cite{grossSchurWeylDualityClifford2021a} proved several 
 results about the operators $R(T)$:

\begin{fact}[Eq. (4.10) in \cite{grossSchurWeylDualityClifford2021a}]
\label{fact:stabilizer-powers-sandwich}
For all $T\in\Sigma_{t,t}$ and all stabilizer states $\ket{S}$, it holds that
\begin{equation}
    \bra{S}^{\otimes t} R(T)\ket{S}^{\otimes t} = 1\,.
\end{equation}
\end{fact}

\begin{fact}[Traces of $R(T)$, Remark 5.1 in \cite{grossSchurWeylDualityClifford2021a}]
\label{fact:traces-of-RT}
Let $T\in\Sigma_{t,t} $ and let $\Delta = \{ (\vec x, \vec x) \,| \, \vec x \in \mathbb{F}_2^t \}$ be the diagonal subspace. Let $l = t - \dim ( T \cap \Delta)$, then 
    \begin{equation}
        \tr \, R(T) = (\tr \, r(T) )^n = 2^{n(t-l)}\,.
    \end{equation}
    Furthermore, 
    \begin{equation}
        \sum_{T\in\Sigma_{t,t}}\tr \,  R(T)  = 2^{nt} (-2^{-n};2)_{t-1} \, ,
    \end{equation}
     where the $q$-Pochhammer symbol $(-2^{-n};2)_{t-1}$ is given by $(-2^{-n};2)_{t-1} = \prod_{k=0}^{t-2} (1+2^{-n+k} )$.

\end{fact}
We note that $\tr\left(R\left(T\right)\right)=2^{nt}$, i.e. $l=0$, only for the identity element $T=e$.

\begin{fact}[Cardinality of $\Sigma_{t,t}$, Theorem 4.10 in \cite{grossSchurWeylDualityClifford2021a}]\label{fact:cardinality-sigma-tt}
\begin{equation}
        \left|\Sigma_{t,t}\right|=\prod_{k=0}^{t-2}\left(2^{k}+1\right)\leq2^{\frac{1}{2}(t^{2}+5t)} \,.
\end{equation}
\end{fact}

\cite{grossSchurWeylDualityClifford2021a} also characterized the commutant further by uncovering an important group structure within. To this end, need the following definition.

\begin{definition}[Stochastic orthogonal group]
The stochastic orthogonal group, denoted $O_t$, is defined as the group of $t\times t$ binary matrices $O$ such that 
\begin{equation}
    O \vec x \cdot O\vec x = \vec x \cdot \vec x \quad \forall \vec x \in \Ft \,.
\end{equation}
\end{definition}
Note that, for any $O\in O_t$, the subspace $T_O = \{(O\vec x, \vec x) \, | \, \vec x \in \Ft \}$ is a stochastic Lagrangian subspace. That is, $T_O \in \Sigma_{t,t}$ for all $O \in O_t$. In the following, we will thus view $O_t$ as a subset of $\Sigma_{t,t}$, i.e., $O_t \subset \Sigma_{t,t}$. We will denote the identity element in $O_t$ by $e$, it corresponds to the diagonal subspace $\Delta = \{ (\vec x, \vec x) \,| \, \vec x \in \mathbb{F}_2^t \}$. Notice also that the symmetric group on $t$ elements, denoted $\mathcal{S}_t$, can be viewed as a subgroup of $O_t$ by considering its matrix representation on $\mathbb{F}_2^t$.

While for $O\in O_t$, the corresponding operators $R(O)$ are unitary, this is not the case for the operators $R(T)$ for $T \in\Sigma_{t,t}\setminus O_t$. Here, we record a bound on the trace-norm of these operators which we will require later.

\begin{fact}[c.f. Lemma 1 in \cite{haferkampEfficientUnitaryDesigns2023}]\label{fact:1-norm-bound-homeopathy}
Let $T\in \Sigma_{t,t}\setminus O_t$, then
    \begin{equation}
        \left\Vert R\left(T\right)\right\Vert _{1}\leq2^{n\left(t-1\right)}\,.
    \end{equation}
\end{fact}

Lastly, similar to \cite{harrowApproximateOrthogonalityPermutation2023}, we want to quantify the orthogonality of the operators $R(T)$ spanning the commutant $\mathrm{Comm}(\Cln, t)$. To this end, we define their corresponding Gram matrix as follows:

\begin{definition}[Gram matrix $G$ corresponding to $\Sigma_{t,t}$]
We define the Gram matrix corresponding to $\{R(T)\}_{T\in\Sigma_{t,t}}$ as the $|\Sigma_{t,t}| \times |\Sigma_{t,t}|$-matrix with entries given by
\begin{equation}
        G^{(n,t)}_{T,T'} := \tr \left( R(T)^{\dagger} R(T') \right) \quad \text{for }T,T'\in\Sigma_{t,t} \,.
\end{equation}
\end{definition}
For convenience, we will often drop the superscript $(n,t)$ on $G^{(n,t)}$. Next, we state a straightforward bound on the off-diagonal row-sums of the Gram matrix:

\begin{fact}[Bound on the off-diagonal row-sum of the Gram matrix]\label{fact:row-sum-bound}
Let $T\in\Sigma_{t,t}$, then
\begin{equation}
     \sum_{T'\in\Sigma_{t,t}\setminus \{ T \}}G_{T,T'} \leq |\Sigma_{t,t}|\cdot 2^{n(t-1)} \leq 2^{n(t-1) + \frac{1}{2}(t^2+ 5t)}.
\end{equation}
\end{fact}
\begin{proof}
    \begin{equation}
    \sum_{T'\in\Sigma_{t,t}\setminus \{ T \}}G_{T,T'} = \sum_{T'\in\Sigma_{t,t}\setminus \{ T \}} \tr \left( R(T)^{\dagger} R(T') \right) = \sum_{T'\in\Sigma_{t,t}\setminus \{ T \}} \tr \left( r(T)^{\dagger} r(T') \right)^n,
    \end{equation}
    and using $r(T) = \sum_{(\vec x, \vec y) \in T} \ket{x}\bra{y}$, we find
    \begin{equation}
        \tr \left( r(T)^{\dagger} r(T') \right) = | T \cap T '| =\begin{cases}
2^{t} & T=T',\\
\leq2^{t-1} & T\neq T'.
\end{cases}
    \end{equation}
    Hence, 
    \begin{equation}
        \sum_{T'\in\Sigma_{t,t}\setminus \{ T \}}G_{T,T'} \leq |\Sigma_{t,t}|\cdot 2^{n(t-1)}.
    \end{equation}
\end{proof}

\section{Single-copy algorithm for stabilizer testing}\label{sec:upper-bound}
In this section, we will prove the following theorem:

\begin{theorem}[Upper bound for single-copy testing]\label{thm:main-result-upper-bound}
    Let $n\geq3, \epsilon>3\cdot2^{-n}$ and let $\ket{\psi}$ be an $n$-qubit pure state. There exists a single-copy algorithm for stabilizer testing an unknown state $\ket\psi$ that uses $t = O\left(n/\epsilon^{2}\right)$ copies of $\ket \psi$, runs in time $O(n^3)$, and succeeds with high probability.
\end{theorem}

This section is organized as follows: In \Cref{ssec:average-spanning-probability}, we introduce the main quantity of interest $\avgspanP\left(\ket\psi\right)$, which we call \textit{average spanning probability}, alongside a single-copy algorithm for estimating $\avgspanP\left(\ket\psi\right)$ to additive precision $\epsilon$. The average spanning probability $\avgspanP\left(\ket\psi\right)$ captures the probability that $K=O(n)$ samples drawn randomly according to the computational difference sampling distribution $r_{C\ket\psi}$ of $C \ket \psi$ span the full space $\Fn$ for $C$ drawn uniformly randomly from the Clifford group. The movation behind this quantity is that the stabilizer states are extremal with respect to this quantity.
In \Cref{ssec:stabilizer-value}, we derive an exact expression for this extremal value of the average spanning probability for stabilizer states. In \Cref{ssec:beyond-stabilizer-states}, we discuss how the relation between $r_{\psi}$ and $p_{\psi}$ from \Cref{lem:relation-r-p} allows us to extend these calculations to non-stabilizer states. Lastly, in \Cref{ssec:bound-on-avg-spanning-prob}, we show that for a state $\ket{\psi}$ that is $\epsilon$-far from all stabilizer states, the value of $\avgspanP\left(\ket\psi\right)$ deviates by at least $\Omega(\epsilon)$ from the extremal value attained by the stabilizer states. This immediately implies the main result \Cref{thm:main-result-upper-bound}, because we have an algorithm to estimate $\avgspanP\left(\ket\psi\right)$ and can efficiently compute the exact extremal stabilizer value.

\subsection{Average spanning probability}
\label{ssec:average-spanning-probability}

\begin{definition}[Spanning probability]
\label{def:spanning-probability}
    Let $\ket{\psi}$ be an $n$-qubit pure state and let $r_{\psi}$ be its computational difference sampling distribution.  Let $\vec v_{1},\dots,\vec v_{K}\sim r_{\psi}$ and consider the event that $\vec v_{1},\dots,\vec v_{K}\in\Fn$ span the full space $\Fn$, in short $\langle \vec v_{1},\dots,\vec v_{K} \rangle =\Fn$. The $K$-\textit{spanning probability} of $\ket{\psi}$, denoted $\spanP \left(\ket{\psi}\right)$ is defined as the probability of this event, i.e., 
    \begin{equation}
\spanP \left(\ket{\psi}\right):=\Pr_{\vec v_{1},\dots, \vec v_{K}\sim r_{\psi}}\left(\langle \vec v_{1},\dots,\vec v_{K} \rangle = \Fn \right)\,.
    \end{equation}
\end{definition}

Note that $\spanP\left(\ket\psi\right)$ depends on the (positive integer) parameter $K$ and we have that $\spanP\left(\ket\psi\right)=0$ for $K<n$. Throughout this work, we will choose $K
$ such that $K\geq n$ but also $K=O(n)$. For this canonical choice, we will sometimes drop the $K$ and refer to $\spanP\left(\ket{\psi}\right)$ simply as the spanning probability of $\ket{\psi}$.

\begin{definition}[Average spanning probability]
\label{def:average-spanning-probability}
    Let $\ket{\psi}$ be an $n$-qubit pure state. The \textit{average spanning probability} $\avgspanP \left(\ket\psi\right)$ of $\ket{\psi}$ is defined as
    \begin{equation}
\avgspanP\left(\ket\psi\right):=\Es{C\sim\Cln}\left[\spanP\left(C\ket{\psi}\right)\right]\,.
    \end{equation}
\end{definition}

The motivation for defining the average spanning probability is that it precisely captures the expectation value of the following simple random process consuming single copies of an $n$-qubit pure state $\ket{\psi}$.

\begin{enumerate}
    \item Draw a uniformly random Clifford $C \sim \Cln$.
    \item Sample $\vec v_{1},\dots,\vec v_{K}\sim r_{C\ket\psi}$ by performing computational difference sampling on $C\ket{\psi}$.
    \item If $\langle \vec v_{1},\dots, \vec v_{K}\rangle=\Fn$, output 1, else output 0.
\end{enumerate}
 
Note that the third step involves computing the rank of the $K\times n$ binary matrix $M$ constructed by taking the samples $\vec v_{1},\dots, \vec v_{K}$ as the rows. We have that $ \langle \vec v_{1},\dots, \vec v_{K} \rangle = \Fn$ is equivalent to $\mathrm{rank}\left(M\right)=n$. By repeating the above 3-step process and averaging the outcomes, one can estimate $\avgspanP(\ket{\psi})$. Hence, we record the following lemma:

\begin{lemma}[Estimating $\avgspanP(\ket{\psi})$ with single copies]
\label{lem:estimation}
    Let $\ket{\psi}$ be an $n$-qubit pure state and let $K\geq n$. Then, there exists a single-copy algorithm that, with probability $1-\delta$, produces an estimate $\hat{r}$ such that 
    \begin{equation}
        \left|\hat{r}-\avgspanP\left(\psi\right)\right|\leq\epsilon.
    \end{equation}
     The algorithm consumes $ O\left(\frac{K\log\left(1/\delta\right)}{\epsilon^{2}}\right)$ copies of $\ket{\psi}$ and runs in time $O\left(Kn^{2}\right)$.
\end{lemma}
\begin{proof}
    The sample complexity follows immediately from Hoeffding's inequality since the output of the above 3-step process takes values in  $\left\{ 0,1\right\}$. The run-time consists of the time to generate a random $n$-qubit Clifford, the time to perform computational difference sampling to generate $K$ samples and the time to compute the rank of a $K\times n$-matrix which can be done by naive Gaussian elimination in time $O\left(Kn^{2}\right)$.
\end{proof}

As we will demonstrate through the course of this section, taking $K=O(n)$ and estimating $\avgspanP(\ket{\psi})$ to additive precision $\epsilon$ is sufficient for stabilizer testing. In particular, the time and sample complexities stated in \Cref{thm:main-result-upper-bound} follow directly from \Cref{lem:estimation}.

\subsection{The stabilizer value of \texorpdfstring{$\avgspanP(\ket\psi)$}{Pketpsi}}\label{ssec:stabilizer-value}

Because of the average over the Clifford group $\Cln$ in \Cref{def:average-spanning-probability}, $\avgspanP(\ket\psi)$ takes the same value for all stabilizer states $\ket S \in \stabn$. We call this value the \textit{stabilizer value}, and define it formally as follows:
\begin{definition}[Stabilizer value]
    Let $\ket S\in\stabn$ be an $n$-qubit pure stabilizer state. We define the stabilizer value of the average spanning probability as
    \begin{equation}
        \avgspanP \left(\mathrm{Stab}\left(n\right)\right):=\avgspanP\left(\ket S\right)\,.
    \end{equation}
\end{definition}

Next, we derive an expression for the stabilizer value in terms of $n,K$ which can be efficiently computed exactly. To this end, we first characterize the $K$-spanning probability for stabilizer states. 

\begin{lemma}[Spanning probability for stabilizer states]\label{lem:spanning-prob-stabilizer-states}
    Let $K\geq n$ and let $\ket S$ be a stabilizer state. Furthermore, let $M_{S}$ be the Lagrangian subspace corresponding to the stabilizer group of $\ket S$. Then, the $K$-spanning probability is given by

    \begin{equation}
        \spanP\left(\ket S\right)=\begin{cases}
\prod_{j=0}^{n-1}\left(1-2^{j-K}\right) & \dim\left(M_{S}\cap\mathcal{Z}\right)=0,\\
0 & \mathrm{otherwise.}
\end{cases}
    \end{equation}
\end{lemma}
\begin{proof}
    If $\dim\left(M_{S}\cap\mathcal{Z}\right)\neq0$, then $r_{\ket{S}}$ is fully confined to a subspace of $\Fn$ and hence the spanning probability is zero.
    Otherwise, if $\dim\left(M_{S}\cap\mathcal{Z}\right)=0$, then $r_{\ket{S}}$ is the uniform distribution over $\Fn$. In that case, the claimed expression is well-known; it is the probability that a random binary $K\times n$-matrix is full rank. 
\end{proof}

With \Cref{lem:spanning-prob-stabilizer-states} established, we can now obtain the stabilizer value for any combination of $K$ and $n$. Recall, that for $K<n$, the spanning probability is always trivially zero, hence the following theorem focuses on the case $K\geq n$. For convenience, we point the reader to the definition of $Q(n,k)$ given in \Cref{def:p-n-k-def}.
\begin{theorem}[Stabilizer value in terms of $n,K$]\label{thm:stabilizer-value}
Let $K\geq n$, then
\begin{equation}
    \avgspanP\left(\stabn \right) =Q\left(n,0\right)\times\prod_{j=0}^{n-1}\left(1-2^{j-K}\right)=\prod_{j=1}^{n}\frac{\left(1-2^{j-1-K}\right)}{\left(1+2^{-j}\right)}.
\end{equation}
    
\end{theorem}

\begin{proof}
    Let $\ket S$ be any n-qubit pure stabilizer state and let $M_{S}$ be the Lagrangian subspace corresponding to the stabilizer group of $\ket S$. We partition the Clifford group as $\Cln = \bigcup_{k=0}^n \Clnk$ where
    \begin{equation}\label{eq:Clifford-partition}
        \Clnk:=\left\{ C\in\Cln:\dim\left(C\left(M_{S}\right)\cap\mathcal{Z}\right)=k\right\} \,.
    \end{equation}
    Then,
    \begin{align}
\Es{C \sim \Cln}\left[\spanP\left(C\ket S\right)\right]
&=\frac{1}{\left|\mathrm{Cl}\left(n\right)\right|}\sum_{k=0}^{n}\sum_{C\in\mathrm{Cl}\left(n,k\right)}\spanP\left(C\ket S\right)\, ,\\
&=\frac{1}{\left|\mathrm{Cl}\left(n\right)\right|}\sum_{C\in\mathrm{Cl}\left(n,0\right)}\bigg[\prod_{j=0}^{n-1}\left(1-2^{j-K}\right)\bigg] \, ,\\
&=\underbrace{\frac{\left|\mathrm{Cl}\left(n,0\right)\right|}{\left|\mathrm{Cl}\left(n\right)\right|}}_{=Q(n,0)}\times\prod_{j=0}^{n-1}\left(1-2^{j-K}\right)\, .
    \end{align}
\end{proof}

\subsection{Going beyond stabilizer states}\label{ssec:beyond-stabilizer-states}
In the previous subsection, we found that the $K$-spanning probability
$\spanP\left(\ket{\psi}\right)$ for a state $\ket{\psi}$
is easy to treat in the case of stabilizer states $\ket S$ because
the distribution $r_{\ket S}$ is either uniform over $\mathbb{F}_{2}^{n}$
or confined to some proper subspace of $\mathbb{F}_{2}^{n}$. However,
how can we obtain $\spanP\left(\ket{\psi}\right)$ for general
pure states $\ket{\psi}$ where $r_{\psi}$ does not take such a simple form?

Instead of calculating $\spanP\left(\ket{\psi}\right)$ exactly,
in this section, we outline our approach for obtaining bounds on $\spanP\left(\ket{\psi}\right)$.
Firstly, note that we can rewrite the $K$-spanning probability as
follows,
\begin{align}
\spanP\left(\ket{\psi}\right) & =1-\Pr_{\vec v_{1}, \dots, \vec v_{K}\sim r_{\psi}}\left(\langle \vec v_{1}, \dots, \vec v_{K} \rangle\neq\mathbb{F}_{2}^{n}\right)\\
 & =1-\Pr_{\vec v_{1}, \dots, \vec v_{K}\sim r_{\psi}}\left(\bigcup_{L\in\mathcal{L}\setminus\left\{ \mathbb{F}_{2}^{n}\right\} }\left( \langle \vec v_{1}, \dots, \vec v_{K}\rangle=L\right)\right)\,,\label{eq:union-over-subspaces}
\end{align}
where the union goes over all subspaces $L\in\mathcal{L}\setminus\left\{ \mathbb{F}_{2}^{n}\right\} $
and $\mathcal{L}$ is the set of all subspaces of $\mathbb{F}_{2}^{n}$.
We can further rewrite this as 
\begin{equation}
\spanP\left(\ket{\psi}\right)=1-\Pr_{\vec v_{1}, \dots, \vec v_{K}\sim r_{\psi}}\left(\bigcup_{L\in\mathcal{L}\setminus\left\{ \mathbb{F}_{2}^{n}\right\} }\left(\langle \vec v_{1}, \dots, \vec v_{K} \rangle\subseteq L\right)\right),\label{eq:union-over-subspaces-contained}
\end{equation}
with the only difference between \Cref{eq:union-over-subspaces}
and \Cref{eq:union-over-subspaces-contained} being that we exchanged
$\langle \vec v_{1}, \dots, \vec v_{K} \rangle=L$ for $\langle \vec v_{1}, \dots, \vec v_{K} \rangle\subseteq L$.
Note that we can restrict the union in \Cref{eq:union-over-subspaces-contained}
to range only over the $\left(n-1\right)$-dimensional subspaces of
$\mathbb{F}_{2}^{n}$ because these contain all other smaller subspaces
in $\mathcal{L}$. So, let us define $\mathcal{L}^{n-1}$ to be the
set of all $2^{n}-1$ many $\left(n-1\right)$-dimensional subspaces
of $\mathbb{F}_{2}^{n}$. Then, we can rewrite $\spanP\left(\ket \psi\right)$
as
\begin{equation}
\spanP\left(\ket\psi\right)=1-\Pr_{\vec v_{1}, \dots, \vec v_{K}\sim r_{\psi}}\left(\bigcup_{L\in\mathcal{L}^{n-1}}\left(\langle \vec v_{1}, \dots, \vec v_{K} \rangle\subseteq L\right)\right).\label{eq:union-over-n-1-subspaces-contained}
\end{equation}
Using a union bound and the general relation between $r_\psi$ and $p_\psi$ from \Cref{lem:relation-r-p}, we find the following lower bound to the spanning
probability:
\begin{lemma}[Lower bound on the spanning probability from union bound]
\label{lem:lower-bound-from-union-bound}
Let $\ket{\psi}$ be an
$n$-qubit pure state. Then, the spanning probability $\spanP\left(\ket\psi\right)$
can be lower bounded as follows,
\begin{align}
\spanP\left( \ket \psi\right) & \geq1-\sum_{\vec y\in\mathcal{Z}^{\times}}\left[\frac{1}{2}+2^{n-1}p_{\psi}\left(\vec y\right)\right]^{K}\,=1-\sum_{\vec y\in\mathcal{Z}^{\times}}\,\left[\frac{1+\trace\left(\psi P_{\vec y}\right)^{2}}{2}\right]^{K}\,.\label{eq:union-bound-spanning-probability}
\end{align}
\end{lemma}

\begin{proof}
Using the union bound, it follows from \Cref{eq:union-over-n-1-subspaces-contained},
that
\begin{align}
\spanP\left(\ket \psi\right) & \geq1-\sum_{L\in\mathcal{L}^{n-1}}\,\Pr_{ \vec v_{1},\dots, \vec v_{K}\sim r_{\psi}}\left(\langle \vec v_{1}, \dots, \vec v_{K} \rangle\subseteq L\right)\\
 & =1-\sum_{L\in\mathcal{L}^{n-1}}\left(r_{\psi}\left(L\right)\right)^{K}
\end{align}
where the equality follows from Lemma \ref{lem:subspace-weight-spanning-probability}.
Now, using the relationship laid out in \Cref{lem:relation-r-p},
we can relate $r_{\psi}\left(L\right)$ to $p_{\psi}\left(\vec 0_n \times L^{\perp}\right)$
as follows,
\begin{align}
r_{\psi}\left(L\right) 
& =\left|L\right|\,p\left(\vec 0_n\times L^{\perp}\right)\\
& =2^{n-1}p\left(\vec 0_n\times L^{\perp}\right)\,.
\end{align}
Note that since $\dim\left(L\right)=n-1$, we have that $\dim\left(L^{\perp}\right)=1$
and so each $L^{\perp}$ is of the form $\langle \vec b \rangle=\left\{ \vec 0_n, \vec b\right\}  $
for some $\vec b\in\mathbb{F}_{2}^{n}\setminus\left\{ \vec 0_n\right\} $
and $\vec 0_n\times L^{\perp}$ is of the form $\langle \vec y \rangle =\left\{ \vec 0_{2n},\vec y\right\} $
where $\vec y=\left(\vec 0_n, \vec b\right)$. Hence, we have that
\begin{equation}
r_{\psi}\left(L\right)=2^{n-1}\left(p_{\psi}\left(\vec 0_{2n} \right)+p_{\psi}\left(\vec y\right)\right)=\frac{1}{2}+2^{n-1}p_{\psi}\left(\vec y\right).
\end{equation}
Finally, when summing over all $L\in\mathcal{L}^{n-1}$, we find
\begin{align}
\spanP\left(\ket \psi\right) & \geq1-\sum_{L\in\mathcal{L}^{n-1}}\left(r_{\psi}\left(L\right)\right)^{K}=1-\sum_{\vec y\in\mathcal{Z}^{\times}}\left[\frac{1}{2}+2^{n-1}p_{\psi}\left(\vec y\right)\right]^{K}\,.\label{eq:lower-bound-in-proof}
\end{align}
\end{proof}

In the next section, we will make use of \Cref{lem:lower-bound-from-union-bound}
in order to lower bound the average spanning probability for arbitrary states .

\subsection{Bound in terms of stabilizer fidelity}\label{ssec:bound-on-avg-spanning-prob}

Below, we will prove our main technical theorem establishing a lower bound on $\avgspanP(\ket \psi)$ in terms of $p_{\psi}$. However, let us briefly remark that the following direct approach fails: Apply
the union bound from \Cref{lem:lower-bound-from-union-bound} and take the average over all Cliffords to obtain
\begin{align}
\avgspanP\left(\ket\psi\right)=\Es{C \sim \Cln}\left[\mathbb{P}_{K}\left(C\ket{\psi}\right)\right] & \geq1-\sum_{\vec y\in\mathcal{Z}^{\times}}\,\Es{C \sim \Cln}\left[\frac{1+\trace\left(\ket{\psi}\bra{\psi}C^{\dagger}P_{\vec y}C\right)^{2}}{2}\right]^{K}\\
 & =1-\frac{2^{n}-1}{4^{n}-1}\sum_{\vec y\in\mathbb{F}_{2}^{2n}\setminus\left\{ \vec 0_{2n} \right\} }\left[\frac{1+\trace\left(\ket{\psi}\bra{\psi}P_{\vec y}\right)^{2}}{2}\right]^{K}.
\end{align}
Here, in the second line, we used that for a random Clifford $C$, the rotated Pauli $C^{\dagger}P_{\vec y}C$
is just a uniformly random Pauli operator.

To see that this bound is not useful, we can e.g. evaluate it for any stabilizer state which, after some calculation, yields an exponentially small lower bound. This bound would not at all capture the stabilizer value which we know is the correct bound from \Cref{ssec:stabilizer-value} for any stabilizer state. The underlying issue is that the union bound from \Cref{lem:lower-bound-from-union-bound} is far too loose to be useful when applied like this.

Instead, to circumvent this issue, our approach is to again first partition the Clifford group similar to \Cref{eq:Clifford-partition}. Crucially, after conditioning on $C$ belonging to $ \mathrm{Cl}\left(n,k\right)$, the union bound actually gives rise to sharp bounds as we will now show.

\begin{theorem}[Lower bound on $\avgspanP(\ket \psi)$
in terms of $p_{\psi}$]\label{thm:main-result-lower-bound}
Let $n\geq3,K\geq5n$
and let $\ket{\psi}$ be an $n$-qubit pure state. Furthermore, let
$M_{\psi}=\{\vec y\in\mathbb{F}_{2}^{2n}:\trace\left(\psi P_{\vec y}\right)^{2}>\frac{1}{2}\} $
and let $M$ be any Lagrangian subspace of $\mathbb{F}_{2}^{2n}$
such that $M_{\psi}\subseteq M$. Note that such an $M$ always exists.
Then,
\begin{equation}
\avgspanP\left(\psi\right)-\avgspanP\left(\mathrm{Stab}\left(n\right)\right)\geq Q\left(n,1\right)\left[1-p_{\psi}\left(M\right)\right]-2^{-n}.
\end{equation}
\end{theorem}

\begin{proof}
 
Fix $M$ to be any Lagrangian subspace of $\mathbb{F}_{2}^{2n}$
such that $M_{\psi}\subseteq M$. Then, similar to the proof of \Cref{thm:stabilizer-value}, we partition the Clifford group as $\Cln = \bigcup_{k=0}^n \Clnk$ where
    \begin{equation}
        \mathrm{Cl}\left(n,k\right):=\left\{ C\in\mathrm{Cl}\left(n\right):\dim\left(C\left(M\right)\cap\mathcal{Z}\right)=k\right\} \,.
    \end{equation}
Then, we can write $\avgspanP\left(\psi\right)$ as
\begin{align}
\avgspanP\left(\psi\right) & =\frac{1}{\left|\mathrm{Cl}\left(n\right)\right|}\sum_{k=0}^{n}\sum_{C\in\mathrm{Cl}\left(n,k\right)}\spanP\left(C\ket{\psi}\right)\\
 & \geq\frac{1}{\left|\mathrm{Cl}\left(n\right)\right|}\sum_{C\in\mathrm{Cl}\left(n,0\right)}\spanP\left(C\ket{\psi}\right)+\frac{1}{\left|\mathrm{Cl}\left(n,1\right)\right|}\sum_{C\in\mathrm{Cl}\left(n,1\right)}\spanP\left(C\ket{\psi}\right)\,.\label{eq:zero-and-one-contributions}
\end{align}
That is, we restrict ourselves to the contributions coming from $k=0,1$.
Note that $\spanP\left(C\ket{\psi}\right)\geq0$ since it
is a probability, thus this restriction gives rise to a lower bound.
We will tackle the contributions from $\mathrm{Cl}\left(n,0\right)$
and from $\mathrm{Cl}\left(n,1\right)$, separately, by lower bounding
$\spanP\left(C\ket{\psi}\right)$ in each case. To obtain
these lower bounds, we can now employ the union bound from \Cref{lem:lower-bound-from-union-bound}. Crucially, we have already conditioned on $C$ belonging to either
$\mathrm{Cl}\left(n,0\right)$ or $\mathrm{Cl}\left(n,1\right)$.
Under this condition, the union bound actually gives rise to sharp
bounds.

Concretely, we will rewrite the union bound from \Cref{lem:lower-bound-from-union-bound}
as follows,
\begin{align}
\mathbb{P}_{K}\left(C\ket{\psi}\right) & \geq1-\sum_{\vec y\in\mathcal{Z}^{\times}}\,\left[\frac{1+\trace\left(C\ket{\psi}\bra{\psi}C^{\dagger}P_{\vec y}\right)^{2}}{2}\right]^{K}=1-\sum_{\vec y\in C^{\dagger}\left(\mathcal{Z}^{\times}\right)}\left[\frac{1+\trace\left(\psi P_{\vec y}\right)^{2}}{2}\right]^{K}\,,\label{eq:lower-bound-c-dagger-z}
\end{align}
where the last equality follows from  \Cref{fact:Clifford-mapping-dagger}.
We also note that
\begin{equation}
C\left(M\right)\cap\mathcal{Z}=M\cap C^{\dagger}\left(\mathcal{Z}\right)\,.
\end{equation}
We now treat the two sets $\mathrm{Cl}\left(n,0\right)$ and $\mathrm{Cl}\left(n,1\right)$, individually. We begin by lower bounding $\mathrm{Cl}\left(n,0\right)$:
\begin{equation}\label{eq:zero-order-contribution}
\frac{1}{\left|\mathrm{Cl}\left(n\right)\right|}\sum_{C\in\mathrm{Cl}\left(n,0\right)}\mathbb{P}_{K}\left(C\ket{\psi}\right)\geq Q\left(n,0\right)\times\left[1-\left(2^{n}-1\right)\left(\frac{3}{4}\right)^{K}\right].
\end{equation}

This can be obtained as follows: By definition, for all\textbf{ $C\in\mathrm{Cl}\left(n,0\right)$}, we have that
$M\cap C^{\dagger}\left(\mathcal{Z}\right)=\left\{ \vec 0_{2n}\right\} $,
i.e., only the identity $I^{\otimes n}$ lies in the intersection.
Thus, we have that $M\cap C^{\dagger}\left(\mathcal{Z}^{\times}\right)=\emptyset$.
Hence, for all $\vec y\in C^{\dagger}\left(\mathcal{Z}^{\times}\right)$,
we have $\trace\left(\psi P_{\vec y}\right)^{2}\leq1/2$ . Thus, we can
bound $\mathbb{P}_{K}\left(C\ket{\psi}\right)$ via \Cref{eq:lower-bound-c-dagger-z}
as
\begin{equation}
\mathbb{P}_{K}\left(C\ket{\psi}\right)\geq1-\left(2^{n}-1\right)\left[\frac{1+\frac{1}{2}}{2}\right]^{K}=1-\left(2^{n}-1\right)\left(\frac{3}{4}\right)^{K}\,.
\end{equation}
Now, for the sum $\sum_{C\in\mathrm{Cl}\left(n,0\right)}$, we find
\begin{equation*}
\frac{1}{\left|\mathrm{Cl}\left(n\right)\right|}\sum_{C\in\mathrm{Cl}\left(n,0\right)}\mathbb{P}_{K}\left(C\ket{\psi}\right)
\geq   \frac{\left|\mathrm{Cl}\left(n, 0\right)\right|}{\left|\mathrm{Cl}\left(n\right)\right|} 
\cdot \left[1-\left(2^{n}-1\right)\left(\frac{3}{4}\right)^{K}\right]
=Q\left(n,0\right)\left[1-\left(2^{n}-1\right)\left(\frac{3}{4}\right)^{K}\right].
\end{equation*}
Next we prove a lower bound for $\mathrm{Cl}\left(n,1\right)$:
\begin{equation}\label{eq:first-order-contribution}
\frac{1}{\left|\mathrm{Cl}\left(n\right)\right|}\sum_{C\in\mathrm{Cl}\left(n,1\right)}\mathbb{P}_{K}\left(C\ket{\psi}\right)\geq Q\left(n,1\right)\left[1-p_{\psi}\left(M\right)-\left(2^{n}-2\right)\left(\frac{3}{4}\right)^{K}-2^{-K}\right].
\end{equation}
This can be obtained as follows: For any $C\in\mathrm{Cl}\left(n,1\right)$, the intersection $M\cap C^{\dagger}\left(\mathcal{Z}\right)$
is 1-dimensional and hence takes the form $M\cap C^{\dagger}\left(\mathcal{Z}\right)=\left\{ \vec 0^{2n}, \vec z\right\} $
for some $\vec z\in M\subset\mathbb{F}_{2}^{2n}$. Using \Cref{eq:lower-bound-c-dagger-z},
we have
\begin{align}
\mathbb{P}_{K}\left(C\ket{\psi}\right) 
& \geq1-\left[\frac{1+\trace\left(\psi P_{\vec z}\right)^{2}}{2}\right]^{K}-\sum_{\vec y\in C^{\dagger}\left(\mathcal{Z}^{\times}\right)\setminus\left\{ \vec z\right\} }\left[\frac{1+\trace\left(\psi P_{\vec y}\right)^{2}}{2}\right]^{K}\,,\\
 & \geq1-\left[\frac{1+\trace\left(\psi P_{\vec z}\right)^{2}}{2}\right]^{K}-\left(2^{n}-2\right)\left(\frac{3}{4}\right)^{K}\,.
\end{align}
To get the second inequality, we used again that all terms in the
second sum over $C^{\dagger}\left(\mathcal{Z}^{\times}\right)\setminus\left\{ \vec z\right\} $
can be upper bounded by $\left(3/4\right)^{K}$ since $\trace\left(\psi P_{\vec y}\right)^{2}\leq1/2$
for all $\vec y\in C^{\dagger}\left(\mathcal{Z}^{\times}\right)\setminus\left\{ \vec z\right\} $.
Now, when summing over all $C\in\mathrm{Cl}\left(n,1\right)$, we have
\begin{equation*}
\frac{1}{\left|\mathrm{Cl}\left(n\right)\right|}\sum_{C\in\mathrm{Cl}\left(n,1\right)}\mathbb{P}_{K}\left(C\ket{\psi}\right)\geq\frac{\left|\mathrm{Cl}\left(n,1\right)\right|}{\left|\mathrm{Cl}\left(n\right)\right|}\left(1-\left(2^{n}-2\right)\left(\frac{3}{4}\right)^{K}-\frac{1}{2^{n}-1}\sum_{\vec z\in M\setminus\left\{ \vec 0_{2n}\right\} }\left[\frac{1+\trace\left(\psi P_{\vec z}\right)^{2}}{2}\right]^{K}\right)\,.
\end{equation*}
We can lower bound this expression by using that $\left(1+x\right)^{r}\leq1+\left(2^{r}-1\right)x$
for $x\in\left[0,1\right],r>1$, which yields
\begin{equation}
\left[\frac{1+\trace\left(\psi P_{\vec z}\right)^{2}}{2}\right]^{K}\leq\trace\left(\psi P_{\vec z}\right)^{2}+2^{-K},
\end{equation}
and so
\begin{align}
\frac{1}{2^{n}-1}\sum_{\vec z\in M\setminus\left\{ \vec 0_{2n}\right\} }\left[\frac{1+\trace\left(\psi P_{\vec z}\right)^{2}}{2}\right]^{K} & \leq\frac{1}{2^{n}-1}\sum_{\vec z\in M\setminus\left\{ \vec 0_{2n}\right\} }\trace\left(\psi P_{z}\right)^{2}+2^{-K}\\
 & \leq\frac{1}{2^{n}}\sum_{\vec z\in M}\trace\left(\psi P_{\vec z}\right)^{2}+2^{-K}\\
 & =p_{\psi}\left(M\right)+2^{-K}.
\end{align}
So, overall, we find
\begin{equation}
\frac{1}{\left|\mathrm{Cl}\left(n\right)\right|}\sum_{C\in\mathrm{Cl}\left(n,1\right)}\mathbb{P}_{K}\left(C\ket{\psi}\right)\geq Q\left(n,1\right)\left[1-p_{\psi}\left(M\right)-\left(2^{n}-2\right)\left(\frac{3}{4}\right)^{K}-2^{-K}\right].
\end{equation}
We can now combine all these facts and finish the proof.  We recall \Cref{eq:zero-and-one-contributions} and plug in
 \Cref{eq:zero-order-contribution} and \Cref{eq:first-order-contribution}
to obtain
\begin{align}
\overline{\mathbb{P}_{K}}\left(\ket\psi\right) & \geq\frac{1}{\left|\mathrm{Cl}\left(n\right)\right|}\sum_{C\in\mathrm{Cl}\left(n,0\right)}\mathbb{P}_{K}\left(C\ket{\psi}\right)+\frac{1}{\left|\mathrm{Cl}\left(n,1\right)\right|}\sum_{C\in\mathrm{Cl}\left(n,1\right)}\mathbb{P}_{K}\left(C\ket{\psi}\right)\\
 & \geq Q\left(n,0\right)\left[1-\left(2^{n}-1\right)\left(\frac{3}{4}\right)^{K}\right]+Q\left(n,1\right)\left[1-p_{\psi}\left(M\right)-\left(2^{n}-2\right)\left(\frac{3}{4}\right)^{K}-2^{-K}\right]\,.
\end{align}
Now, recall also that the stabilizer value is upper bounded by $\overline{\mathbb{P}_{K}}\left(\mathrm{Stab}\left(n\right)\right)<Q\left(n,0\right)$.
Hence, we obtain
\begin{align}
\overline{\mathbb{P}_{K}}\left(\ket\psi\right)-\overline{\mathbb{P}_{K}}\left(\mathrm{Stab}\left(n\right)\right) & >Q\left(n,1\right)\left[1-p_{\psi}\left(M\right)\right]-\left(\frac{3}{4}\right)^{K}Q\left(n,0\right)\left(2^{n}-1\right)\left[1+\frac{\left(2^{n}-2+2^{-K}\right)}{2^{n}}\right]\notag\\
 & \geq Q\left(n,1\right)\left[1-p_{\psi}\left(M\right)\right]-\left(\frac{3}{4}\right)^{K}2^{n}\cdot2Q\left(n,0\right),
\end{align}
where we used that 
\begin{equation}
Q\left(n,1\right)=Q\left(n,0\right)\frac{2^{n}-1}{2^{n}}.
\end{equation}
For $n\geq3$, $2\,Q\left(n,0\right)<1$ and for any $K=cn$ for some
integer $c$ with $c>\log_{4/3}\left(4\right)\approx4.82$, we have
$\left(\frac{3}{4}\right)^{K}2^{n}\leq2^{-n}$. So, for instance,
we can choose $c=5$. Then, for $n\geq3$, we have
\begin{equation}
\left(\frac{3}{4}\right)^{K}2^{n}\cdot2Q\left(n,0\right)\leq\frac{1}{2^{n}}\,,
\end{equation}
and we obtain our final bound
\begin{equation}
\overline{\mathbb{P}_{K}}\left(\ket\psi\right)-\overline{\mathbb{P}_{K}}\left(\mathrm{Stab}\left(n\right)\right)\geq Q\left(n,1\right)\left[1-p_{\psi}\left(M\right)\right]-2^{-n}\,.
\end{equation}
\end{proof}
From \Cref{fact:lower-bound-stabilizer-fidelity},
we have that for any Lagrangian subspace $M$, the stabilizer fidelity
is lower bounded via $F_{\mathrm{Stab}}\left(\psi\right)\geq p_{\psi}\left(M\right)$.
Hence, we immediately obtain the following corollary:
\begin{corollary}[Difference from stabilizer value]
\label{cor:difference-from-stabilizer-value}
Let $n\geq3,K\geq5n$
and let $\ket{\psi}$ be an $n$-qubit pure state with stabilizer
fidelity $F_{\mathrm{Stab}}\left(\ket\psi\right)$, then
\begin{equation}
\avgspanP\left(\ket\psi\right)-\avgspanP\left(\mathrm{Stab}\left(n\right)\right)\geq Q\left(n,1\right)\left[1-F_{\mathrm{Stab}}\left(\ket\psi\right)\right]-2^{-n}.
\end{equation}
\end{corollary}
Recall from \Cref{cor:p-n-k-values} that $Q(n,1)$ is a large constant. Hence, \Cref{cor:difference-from-stabilizer-value} implies that for state $\ket{\psi}$ with stabilizer infidelity of at least $\epsilon$, $\avgspanP\left(\ket\psi\right)$ differs from the stabilizer value by $\Omega(\epsilon)$. This difference will hence be detected upon estimating $\avgspanP\left(\ket\psi\right)$ to precision $\epsilon$. Combining this insight with \Cref{lem:estimation} yields \Cref{thm:main-result-upper-bound}.

\section{Lower bound for single-copy stabilizer testing}\label{sec:lower-bound}

In this section, we will prove the following theorem:
\begin{theorem}[Lower bound for single-copy stabilizer testing]\label{thm:stabilizer-testing-lower-bound}
    Any single-copy algorithm for stabilizer testing to accuracy $0<\epsilon < 1-n^2/2^{n}$ requires at least $t=\Omega(\sqrt{n})$ copies.
\end{theorem}

This section is organized as follows: Firstly, in \Cref{ssec:reduction-to-maximally-mixed}, we will argue that the sample complexity lower bound stated in \Cref{thm:stabilizer-testing-lower-bound} can be derived from the many-versus-one distinguishing task of distinguishing uniformly random stabilizer states from the maximally mixed state. This task is analogous to the reduction of purity testing to the task of distinguishing Haar random states from the maximally mixed state, which was previously considered in \cite{chenExponentialSeparationsLearning2022a,harrowApproximateOrthogonalityPermutation2023,chenOptimalTradeoffsEstimating2024a}. The main difference is that the unitary group is replaced with the Clifford group. Based on this observation, in \Cref{ssec:bounding-the-bias}, we adapt the proof strategy of \cite{harrowApproximateOrthogonalityPermutation2023} to the Clifford group. In particular, we use the fact that the operators spanning the commutant $\mathrm{Comm}(\Cln, t)$ are approximately orthogonal with respect to the Hilbert-Schmidt inner product in the regime where $t\ll n$ and combine this with new insights about partial transposes of the operators $R(O)$ for $O\in O_t$ to derive lower bounds on distinguishing uniformly random states from the maximally mixed state. 
In \Cref{ssec:partial-transposes}, we collect our results about partial transposes of $R(O)$ for $O\in O_t$. This subsection is modular and might be of independent interest.
Finally, in \Cref{ssec:tree-representation-framework}, we comment on an alternative route to establishing a single-copy sample complexity lower bound via the tree representation framework (c.f. e.g. \cite{chenExponentialSeparationsLearning2022a}).

\subsection{Reduction to distinguishing random stabilizer states from the maximally mixed state}\label{ssec:reduction-to-maximally-mixed}

Our starting point is Le Cam's two-point method, that is, we consider distinguishing tasks between ensembles of states:

\begin{definition}[Distinguishing task from $t$ copies]
Let $\mu$ and $\nu$ be two ensembles of $n$-qubit quantum states. We consider the following two events to
happen with equal probability of 1/2: 
    \begin{itemize}
        \item The unknown state $\rho$ is sampled according to $\mu$.
        \item The unknown state $\rho$ is sampled according to $\nu$.
    \end{itemize}
   Given $t$ copies of $\rho$, the goal is to give a quantum algorithm that correctly distinguishes between these two events with probability $\geq 2/3$.
\end{definition}
 
Throughout, we fix the number of copies to be $t$ so that the algorithm has access to $\Ex[\rho^{\otimes t}]$ where the expectation is taken with respect to one of the ensembles $\mu,\nu$. For the purpose of this work, we consider the following three ensembles of states and corresponding states that the algorithm has access to:
\begin{enumerate}[label=\textbf{(\alph*)}]
    \item[(H)] Haar random $n$-qubit states, $\rho_H := \Es{\ket{\psi} \sim \mu_H}\left[\ket{\psi}\!\bra{\psi}^{\otimes{t}}\right]$
    \item[(S)] uniformly random $n$-qubit stabilizer states, $\rho_S := \Es{\ket{S} \sim \stabn}\left[\ket{S}\!\bra{S}^{\otimes{t}} \right]$
    \item[(I)] the maximally mixed $n$-qubit state,  $\rho_I := \mathds{1}^{\otimes t}/2^{nt}$.
\end{enumerate}
For each pair of ensembles, we can consider an associated distinguishing task. For instance, the pair $(H, S)$ corresponds to distinguishing Haar random states from uniformly random stabilizer states. This task is also the natural starting point for proving a lower bound on the sample complexity of stabilizer testing since it can be reduced to stabilizer testing. This observation is formalized via the following lemma:
\begin{lemma}\label{lem:reduction-to-haar-random-vs-stabilizer}
    Let $0<\epsilon < 1-n^2/2^{n}$. Then, any algorithm for stabilizer testing to accuracy $\epsilon$ can solve the distinguishing task of Haar random states versus uniformly random stabilizer states with probability $1-2^{-O(n^2)}$.
\end{lemma}
\Cref{lem:reduction-to-haar-random-vs-stabilizer} is clear except for the subtlety, that the Haar random ensemble and the random stabilizer state ensemble overlap, i.e., when drawing a Haar random state, with a non-zero probability, the sampled state is $\epsilon$-close to a stabilizer state. Thus, to obtain \Cref{lem:reduction-to-haar-random-vs-stabilizer}, we bound this probability as follows:
\begin{fact} Let $0<\epsilon < 1- n^2/2^{n}$. Then,
    \begin{equation}
        \Pr_{\ket{\psi}}\left[\max_{\ket S\in\mathrm{Stab}\left(n\right)}\left|\braket{S|\psi}\right|^{2}\geq1-\epsilon\right]\leq 1-2^{-O(n^2)}\,.
    \end{equation}
\end{fact}
\begin{proof}
    Let $\ket{\phi}\in\mathbb{C}^{2^n}$ be a fixed $n$-qubit pure state. Then, the probability of a random state $\ket{\psi}$ being $\epsilon$-close in fidelity to the fixed state $\ket{\phi}$ is bounded as follows, see e.g. \cite[Example~55]{meleIntroductionHaarMeasure2024}: 
\begin{equation}
    \Pr_{\ket{\psi}}\left[\left|\braket{\phi|\psi}\right|^{2}\geq1-\epsilon\right]\leq2\exp\left(-\frac{2^n}{2}\left(1-\epsilon\right)\right)\,.
\end{equation}
The number of stabilizer states is $\left|\mathrm{Stab}\left(n\right)\right|=2^{n}\prod_{i=1}^{n}\left(2^{i}+1\right) \leq 2^{\frac{1}{2}\left(n^{2}+5n\right)}$. The result now follows from the union bound.
\end{proof} 

Next, we argue that, when considering single-copy algorithms, any sample complexity lower bound for testing uniformly random stabilizer states versus the maximally mixed state (the pair $(S, I)$) leads to a lower bound for the pair $(H, S)$. 
This essentially follows from a triangle inequality between the three pairs as we now explain: Consider a single-copy distinguishing algorithm for the pair $(H,S)$ using $t$ copies of the unknown state $\rho$ guessing $0$ or $1$, corresponding to the case $(H)$ and $(S)$, respectively. We can model this distinguishing algorithm via a two-outcome POVM $\{M_0, M_1\}$ acting on $\rho_H$ or $\rho_S$, respectively. Letting $M=M_0-M_1$, the bias achieved by the algorithm is given by $|\tr \left(M \left(\rho_H - \rho_S\right)\right)|$. In order to correctly distinguish between case $(H)$ and $(S)$ with probability $\geq 2/3$, the bias needs to be at least $2/3$. Our goal is to upper bound the bias $ |\tr \left(M \left(\rho_H - \rho_S\right)\right)|$ in terms of the number of copies $t$.

By the triangle inequality, we have
\begin{equation}
\label{eq:abc-triangle-inequa}
    |\tr \left(M \left(\rho_H - \rho_S\right)\right)| \leq | \tr \left(M \left(\rho_H - \rho_I\right)\right)| +  |\tr \left(M \left(\rho_S - \rho_I\right)\right)| \,.
\end{equation}
For the pair $(H, I)$, it is shown in \cite{harrowApproximateOrthogonalityPermutation2023} that any single-copy algorithm using at most $t$ copies of $\rho$ can achieve a bias of at most $| \tr \left(M \left(\rho_H - \rho_I\right)\right)|=O\left(t^2\,2^{-n/2}\right)$. For the regime of $t\leq n$, this contribution to the RHS is hence exponentially small in $n$. As we will show $|\tr \left(M \left(\rho_S - \rho_I\right)\right)|$ will be the dominant contribution to the RHS.  Hence, for our purposes, it suffices to focus on bounding the bias for the pair $(S,I)$, corresponding to distinguishing a uniformly random stabilizer state from the maximally mixed state.

We note that \cite{chenExponentialSeparationsLearning2022a} proved an even stronger lower bound of $t=\Omega(2^{n/2})$ on the number of copies for distinguishing Haar random states from the maximally mixed state (the pair $(H,I)$). This translates to a bound  $| \tr \left(M \left(\rho_H - \rho_I\right)\right)|= O\left(t \,2^{-n/2}\right)$ on the bias.

\subsection{Lower bound via the commutant of the Clifford tensor action}\label{ssec:bounding-the-bias}
In this section, we focus on bounding the bias $|\tr \left(M \left(\rho_S - \rho_I\right)\right)| $ where $\rho_S = \Es{\ket{S} }\left[\ket{S}\bra{S}^{\otimes{t}} \right]$ and $\rho_I = \mathds{1}^{\otimes t}/2^{nt}$. That is we focus on the task of distinguishing uniformly random stabilizer states from the maximally mixed state given $t$ copies of the unknown state. Our proof strategy closely resembles that of Ref.~\cite{harrowApproximateOrthogonalityPermutation2023}. The key observation that we will use is that any measurement POVM $\{M_0,M_1\}$ implementable by an LOCC protocol \cite{chitambarEverythingYouAlways2014} (corresponding to a single-copy measurement strategy), is also \emph{Positive Partial Transpose} (PPT), i.e.
\begin{equation}
     0\preceq M_i^{\Gamma_{S}}\preceq I\quad\forall S\subseteq[t]\,.
\end{equation}
This implies that the difference $M=M_0 - M_1$ satisfies $-I\preceq M^{\Gamma_S}\preceq I$. Here, the superscript $^{\Gamma_S}$ denotes taking a partial transpose with respect to a subset $S\subset [t]$. Furthermore, note that $\rho_S,\rho_I$ both commute with all $C^{\otimes t}$ for all $C\in \Cln$, i.e., $\rho_S,\rho_I \in \mathrm{Comm}(\Cln, t) $. The second key insight is that, hence, WLOG we can restrict ourselves to consider $M$ such that $\left[M,C^{\otimes t}\right]=0$ as well. To see this, consider 
\begin{equation}
    \trace\left(M\,\left(\rho_S - \rho_I\right)\right)
    =\Es{C} \left[\tr\left( MC^{\otimes t}\left(\rho_S - \rho_I\right)C^{\dagger \otimes t}\right)\right]
    = \trace\left(  \Es{C}\left[C^{\dagger \otimes t} MC^{\otimes t}\right]\:\left(\rho_S - \rho_I\right)\right)\,,
\end{equation}
where $\Es{C}\left[C^{\dagger \otimes t}(\cdot)C^{\otimes t}\right]$ projects onto the commutant $\mathrm{Comm}(\Cln, t)$. Hence, we will henceforth assume that $M\in \mathrm{Comm}(\Cln , t)$ and so $M$ can be expanded as $M=\sum_{T\in\Sigma_{t,t}}m_{T}R\left(T\right)$, where $R(T)$ are the operators spanning the commutant of the Clifford group, as discussed in \cref{ssec:commutant}.\\

We will now begin by showing an auxiliary bound on $\left|\trace\left(R\left(T\right) M\right)\right|$ for all $T\neq e$. To do so, we make use of an upper bound on $\| R\left(O\right)^{\Gamma_{S}}\| _{1} \leq 2^{n(t-1)}$ established in \Cref{cor:trace-norm-bound-partial-transpose} in the next subsection, namely \Cref{ssec:partial-transposes}. This is ultimately also how we make use of the PPT constraint on the measurement.
\begin{lemma}[auxiliary result]\label{lem:auxiliary-upper-bound}
  Let $M=\sum_{T\in\Sigma_{t,t}}m_{T} \,R(T)$ be such that it satisfies the  condition
    \begin{equation}\label{eq:ppt-constraint}
         -I\preceq M^{\Gamma_{S}}\preceq I\quad\forall S\subseteq[t]\,.
    \end{equation}
    Then, for all $T\in\Sigma_{t,t}$ such that $T\neq e$,
    \begin{equation}
\left|\trace\left(R\left(T\right) M\right)\right|
\leq2^{n\left(t-1\right)} \,.
    \end{equation}
\end{lemma}
\begin{proof}
We will treat $O_t\setminus\{e \}$ and the remaining part $\Sigma_{t,t}\setminus O_t$, separately. For $O_t\setminus\{e \}$, we will use the PPT constraint from \Cref{eq:ppt-constraint}. Using Hölder's inequality, we find for all $O \in O_t\setminus\{e \}$
    \begin{align}
\left|\trace\left(R\left(O\right)M\right)\right| 
& = \left|\trace\left(R\left(O\right)^{\Gamma_{S}}M^{\Gamma_{S}}\right)\right|\leq\left\Vert R\left(O\right)^{\Gamma_{S}}M^{\Gamma_{S}}\right\Vert _{1} \, ,\\
&\leq\left\Vert R\left(O\right)^{\Gamma_{S}}\right\Vert _{1}\left\Vert M^{\Gamma_{S}} \right\Vert _{\infty} \, ,\\
&\leq\left\Vert R\left(O\right)^{\Gamma_{S}}\right\Vert _{1} \, ,\\
&\leq2^{n\left(t-1\right)} \,.
    \end{align}
    where in the last line we used \Cref{cor:trace-norm-bound-partial-transpose}.

    Similarly, for all $T\in\Sigma_{t,t}\setminus O_t$, we can use $-I \preceq  M \preceq I$ and  \Cref{fact:1-norm-bound-homeopathy} to find
    \begin{equation}
\left|\trace\left(R\left(T\right)M\right)\right|\leq\left\Vert R\left(T\right)\right\Vert _{1}\left\Vert M\right\Vert _{\infty}\leq\left\Vert R\left(T\right)\right\Vert _{1}\leq2^{n\left(t-1\right)} \,.
    \end{equation}
\end{proof}

Next, we show how to establish a bound on the magnitudes $|m_T|$ for all $T\neq e$ when $M$ is \emph{traceless} and PPT:

\begin{lemma}[Upper bound on $|m_T|$]\label{lem:upper-bound-coefficients}
Let $M=\sum_{T\in\Sigma_{t,t}}m_{T} \,R(T)$ be such that $\tr(M)=0$ and it satisfies the condition
    \begin{equation}
         -I\preceq M^{\Gamma_{S}}\preceq I\quad\forall S\subseteq[t]\,.
    \end{equation} 
    Provided that $\tfrac{1}{2}(t^2  + 5t) + 1 \leq n$ we have that  $|m_T|\leq 2^{-n+1}$ for all $T\in\Sigma_{t,t}\setminus \{e\}$.
    
\end{lemma}

\begin{proof}
    Let 
    $T= \underset{T'\in\Sigma_{t,t}}{\arg \max} |m_{T'}|$.
Then,
     \begin{align}
\left|\trace\left(R\left(T\right)^{\dagger}M\right)\right|
&=\left|\sum_{T'\in\Sigma_{t,t}}m_{T'}\,\trace\left(R\left(T\right)^{\dagger}R\left(T'\right)\right)\right| \\
&\geq2^{nt}\left|m_{T}\right|-\left|\sum_{T' \in\Sigma_{t,t},T'\neq T}m_{T'}\, \trace\left(R\left(T\right)^{\dagger}R\left(T'\right)\right)\right| \\
&\geq 2^{nt}\left|m_{T}\right|-\sum_{T'\in\Sigma_{t,t}, T'\neq T} 
 \underbrace{|m_{T'}|}_{\leq |m_{T}|} \cdot \underbrace{\left|\trace\left(R\left(T\right)^{\dagger}R\left(T'\right)\right)\right|}_{=G_{T,T'}} \\
&\geq 2^{nt}\left|m_{T}\right|\left(1-2^{-nt}\sum_{T'\in\Sigma_{t,t}, T'\neq T}G_{T,T'}\right)\\
&\geq 2^{nt} \, |m_{T}| \, \left(1 - 2^{-n+\frac{1}{2}(t^2 +5t)} \right) \,.
   \end{align}
Here, in the last step, we plugged in the bound on the row-sum of the Gram matrix from \Cref{fact:row-sum-bound}.
Now consider the case where the ${\arg \max}$ is the identity, i.e., $T=e$. In this case we have from the assumption $\tr(M)=0$ that $0 \geq 2^{nt} \, |m_{e}| \, (1 - 2^{-n+\frac{1}{2}(t^2 +5t)})$. For $\tfrac{1}{2}(t^2 +5t) + 1 \leq n$ this can only be satisfied when $m_e = 0$. Alternatively, when $T\neq e$  we can use the upper bound from \Cref{lem:auxiliary-upper-bound} to find
\begin{equation}
       2^{nt} \, |m_T| \, \left(1 - 2^{-n+\frac{1}{2}(t^2 +5t)} \right) \leq \left|\trace\left(R\left(T\right)^{\dagger}M\right)\right| \leq  2^{n(t-1)}\,. 
\end{equation}
Now using that $\frac{1}{2}(t^2 +5t) + 1 \leq n$ we have
\begin{equation}
    \left|m_{T}\right|\leq 2^{-n} \left(1 - 2^{-n+\frac{1}{2}(t^2 +5t)} \right)^{-1} \leq 2^{-n+1} \,.
\end{equation}
\end{proof}
We are now in the position to prove our main result of this section, the bound on the bias $ |\tr \left(M \left(\rho_S - \rho_I\right)\right)|$.
\begin{theorem}[Bound on bias]
\label{thm:bias-bound-main-theorem}
Let $\{M_0, M_1\}$ be a PPT measurement, let $M=M_0-M_1$ and $\frac{1}{2}(t^2 +5t) + 2 \leq n/2$, and let $\rho_S = \Es{\ket{S} }\left[\ket{S}\bra{S}^{\otimes{t}} \right]$ and $\rho_I=\mathds{1}^{\otimes t}/2^{nt}$, then
    \begin{equation}
         |\tr \left(M \left(\rho_S - \rho_I\right)\right)|  
         \leq  2^{-n/2}\,.
    \end{equation}
\end{theorem}
\begin{proof}
Without loss of generality we assume that $M \in \mathrm{Comm}(\Cln, t)$ so that we can write $M=\sum_{T\in\Sigma_{t,t}}m_{T}R\left(T\right)$. We begin by separating out the traceless part of $M$ by writing $M = 2^{-nt}\,\tr(M)I + M'$ with $\tr(M') = 0$. Note that by the definition of $M$ in terms of the POVM elements $M_0,M_1$, we have $|2^{-nt}\, \tr (M)|\leq 1$ and so using the triangle inequality we have that $-I \preceq M'/2 \preceq I$. Using that $\tr(\rho_S-\rho_I)=0$ and the triangle inequality, we have
    \begin{equation}
        |\tr \left(M \left(\rho_S - \rho_I\right)\right)| =|\tr \left(M
        '\left(\rho_S - \rho_I\right)\right)|  \leq \sum_{T\in\Sigma_{t,t}, T\neq e} |m'_{T}|\cdot |\trace\left(R\left(T\right)\left(\rho_S - \rho_I\right)\right)|\,,
    \end{equation}
    where $M' = \sum_{T\in \Sigma_{t,t}}m'_T R(T)$ defines the parameters $m'_T$.
    Plugging in the expressions for $\rho_S, \rho_I$, we find that
    \begin{align}
    \trace\left(R\left(T\right) \rho_S \right) 
        &=  \Es{\ket{S} \sim \stabn}\left[ \trace\left(R\left(T\right)\ket{S}\bra{S}^{\otimes{t}} \right) \right] =1 \, ,\\
        \trace\left(R\left(T\right) \rho_I \right) 
        &=\frac{\trace\left(R\left(T\right)\right)}{2^{nt}} \leq 2^{-n} \qquad \forall T\in\Sigma_{t,t}, T\neq e \, .
    \end{align}
        In the first line, we used that $\bra{S}^{\otimes t} R(T)\ket{S}^{\otimes t} = 1$ for all $T\in\Sigma_{t,t}$ (see \Cref{fact:stabilizer-powers-sandwich}). In the second line, we used that $\tr\left(R\left(T\right)\right)=2^{n\left(t-l\right)}$ with $l=0$ only for the identity element $T=e$ (see \Cref{fact:traces-of-RT}).
   It follows that 
   \begin{equation}
       |\tr \left(M \left(\rho_S - \rho_I\right)\right)| \leq \sum_{T\in\Sigma_{t,t}, T\neq e} |m'_{T}| \cdot |1-2^{-n}| \leq \sum_{T\in\Sigma_{t,t}, T\neq e} |m'_{T}|\,.
   \end{equation}
Now we can apply \Cref{lem:upper-bound-coefficients} to $M'/2$. This gives
    \begin{equation}
        |\tr \left(M \left(\rho_S - \rho_I\right)\right)|  \leq \sum_{T\in\Sigma_{t,t}, T\neq e} |m'_{T}| \leq  2^{-n+2}\cdot |\Sigma_{t,t}| \leq 2^{-n+2 + \frac{1}{2}( t^2+5t)}\leq 2^{-n/2},
    \end{equation}
    by the assumption $\tfrac{1}{2}(t^2 +5t)+2  \leq n/2$. This finishes the proof. 
\end{proof}

As we explained in \Cref{ssec:reduction-to-maximally-mixed}, this bound on the bias immediately implies the following corollary:
\begin{corollary}[Single-copy lower bound for random stabilizer vs. maximally mixed]
    Any single-copy algorithm for distinguishing the maximally mixed state $\mathds{1}/2^n$ from  random $n$-qubit stabilizer states, with probability at least 2/3, requires at least $t=\Omega(\sqrt{n})$ many copies of the unknown state.
\end{corollary}

Furthermore, by virtue of \Cref{lem:reduction-to-haar-random-vs-stabilizer} and the triangle inequality in \Cref{eq:abc-triangle-inequa} between the three ensembles $(H,S,I)$, \Cref{thm:bias-bound-main-theorem} also implies our main result, \Cref{thm:stabilizer-testing-lower-bound}.

\subsection{Partial transposes of \texorpdfstring{$R(O)$}{R(O)}}\label{ssec:partial-transposes}
In this section, we study partial transposes of $R(O)$ for $O\in O_t$. Let $S\subset [t]$, then $S$ and its complement $\bar{S}$ in $[t]$ form a partition of $[t]$. The operators $R(O)$ act on $t$ subsystems. We consider taking partial transposes with respect to a subset $S\subset [t]$ of subsystems and will denote this operation by $^{\Gamma_S}$. Concretely, for $S\subset [t]$, we denote by $R(O)^{\Gamma_S}$ its partial transpose with the respect to the subsystems indexed by $S$.

We start our exposition with a simple relationship of kernels of principal submatrices of binary orthogonal matrices. Recall that for $S,S'\subseteq [t]$, $O_{S, S'}$ denotes the matrix obtained by keeping only the rows with indices in $S$ and columns with indices in $S'$. In particular, $O_{S,S}$ denotes a $|S| \times |S|$- principal submatrix of $O$. 
\begin{lemma}[Kernels of principal submatrices of $O\in O_t$]\label{lem:kernels-of-principal-submatrices}
    For any $O\in O_{t}$ and any $S\subset \left[t\right]$, $\dim\left(\ker\,O_{S,S}\right)=\dim\left(\ker\,O_{\bar{S},\bar{S}}\right)$.
\end{lemma}
\begin{proof}
    We will only use the orthogonality of $O$, namely that $O^T O = I_t$. In block form, this reads
    \begin{equation}
    \begin{pmatrix}O_{S,S}^{T} & O_{\bar{S},S}^{T}\\
O_{S,\bar{S}}^{T} & O_{\bar{S},\bar{S}}^{T}
\end{pmatrix}\begin{pmatrix}O_{S,S} & O_{S,\bar{S}}\\
O_{\bar{S},S} & O_{\bar{S},\bar{S}}
\end{pmatrix}=\begin{pmatrix}I_{S,S} & 0\\
0 & I_{\bar{S},\bar{S}}.
\end{pmatrix}
    \end{equation}
    In total, this block matrix equation gives rise to four equations. The first pair of equations, we will examine is the following:
    \begin{align}
        O_{S,S}^{T}O_{S,S}+O_{\bar{S},S}^{T}O_{\bar{S},S}
        &=I_{S,S}, \\
        O_{S,\bar{S}}^{T}O_{S,S}
        &=O_{\bar{S},\bar{S}}^{T}O_{\bar{S},S} .
    \end{align}
    Suppose there exists a vector $\vec v\neq \vec 0$ such that $\vec v\in \ker O_{S,S} $, then by the second equation, $O_{\bar{S},\bar{S}}^{T}O_{\bar{S},S}\vec v = \vec 0$ and so $O_{\bar{S},S} \vec v\in \ker O_{\bar{S},\bar{S}}^{T}$. Then, the first equation implies that $O_{\bar{S},S}^{T}O_{\bar{S},S} \vec v=\vec v \neq \vec 0$, and so $O_{\bar{S},S}\vec v \neq \vec 0$. Hence,  $|\ker O_{\bar{S},\bar{S}}^{T}|=|\ker O_{\bar{S},\bar{S}}|\geq|\ker O_{S,S}|$.
    The remaining two equations are actually the same equations just with $S$ and $\bar{S}$ exchanged. By the same logic, they imply $|\ker O_{\bar{S},\bar{S}}| \leq |\ker O_{S,S}|$. Hence, we conclude that the cardinalities $|\ker O_{\bar{S},\bar{S}}|$ and $|\ker O_{S,S}|$ must actually be equal. Finally, $\ker O_{S,S}$ is a subspace of $\Ft$ so that $|\ker O_{S,S}| = 2^{\dim \left( \ker O_{S,S} \right)}$.
    
\end{proof}

The next lemma is a generalization of \cite[Lemma 4]{harrowApproximateOrthogonalityPermutation2023} from the permutation operators to the operators $R(O)$ (which include the permutation operators since $\mathcal{S}_t \subseteq O_t$).

\begin{lemma}[Singular values of partial transposes of $R(O)$]\label{lem:singular-values-partial-transpose}
    Choose a set $S\subset [t]$. For any $O\in O_{t}$, let $k=\dim\left(\ker\,O_{S,S}\right)$. Then $R\left(O\right)^{\Gamma_{S}}$ has $2^{n(t-2k)}$ non-zero singular values each equal to $2^{kn}$. Consequently we have
    \begin{equation}
        \left\Vert R\left(O\right)^{\Gamma_{S}}\right\Vert _{1}=2^{n(t-k)}\,.
    \end{equation}
\end{lemma}
\begin{proof}
    Recall that $R(O) = r(O)^{\otimes n}$ and so it suffices to prove the claim for $n=1$. Let $O\in O_{t}$, $S\subset [t]$ and consider the partial transpose $r(O)^{\Gamma_S}$. For simplicity of notation, in the following, we will arrange the subsystems so that those with indices in $S$ appear first, followed by those with indices in $\bar{S}$.  Since $r(O)=\sum_{\vec x}\ket{O \vec x}\bra{\vec x}$, its partial transpose $r(O)^{\Gamma_S}$ is given by
    \begin{equation}
        r(O)^{\Gamma_S}  = \sum_{\vec x \in \Ft}\ket{\vec x_{S}}\bra{\left(O \vec x\right)_{S}}\otimes\ket{\left(O \vec x\right)_{\bar{S}}}\bra{\vec x_{\bar{S}}}\,.
    \end{equation}
    Here, $\vec x_{S}$ denotes the vector of length $|S|$ obtained from $\vec x$ by keeping only the elements whose corresponding indices are in $S$.
    Since we are interested in the singular values of $r(O)^{\Gamma_S} $, we examine $X:=\left(r\left(O\right)^{\Gamma_{S}}\right)^{\dagger}r\left(O\right)^{\Gamma_{S}}$ which is given by 
\begin{equation}
X:=\left(r\left(O\right)^{\Gamma_{S}}\right)^{\dagger}r\left(O\right)^{\Gamma_{S}}
= \sum_{\vec x, \vec y}\ket{\vec x_{S}}\bra{\vec y_{S}}\braket{\left(O \vec x\right)_{S} | \left(O \vec y\right)_{S}}\otimes\ket{\left(O \vec x\right)_{\bar{S}}}\bra{\left(O \vec y\right)_{\bar{S}}}\braket{\vec x_{\bar{S}}| \vec y_{\bar{S}}}.
\end{equation}
We have
\begin{align}
    \left(O \vec x\right)_{S}
    &=O_{S,S}\vec x_{S}+O_{S,\bar{S}} \vec x_{\bar{S}}, \\
    \left(O \vec x\right)_{\bar{S}}
    &=O_{\bar{S},S}\vec x_{S}+O_{\bar{S},\bar{S}} \vec x_{\bar{S}},
\end{align}
and similarly for $\vec y$. Hence, the overlap $\braket{\left(O \vec x\right)_{S} | \left(O\vec y\right)_{S}}$ is non-zero only if
\begin{equation}
     O_{S,S}\vec x_{S}+O_{S,\bar{S}} \vec x_{\bar{S}} = O_{S,S}\vec y_{S}+O_{S,\bar{S}} \vec y_{\bar{S}}\,.
\end{equation}
Taking into account that $\braket{\vec x_{\bar{S}}| \vec y_{\bar{S}}}$ enforces $\vec x_{\bar{S}} = \vec y_{\bar{S}}$, we hence find that $\braket{\left(O \vec x\right)_{S} | \left(O\vec y\right)_{S}}$ is non-zero only if
$\vec x_{S}+\vec y_{S}\in \ker O_{S,S}$.
Hence, 
\begin{equation}
X =
\sum_{\vec x_{S}}\sum_{\vec z_{S}\in \ker O_{S,S}} \ket{\vec x_{S}}\bra{\vec x_{S}+\vec z_{S}}  \otimes\sum_{\vec x_{\bar{S}}}\ket{O_{\bar{S},S}\vec x_{S}+O_{\bar{S},\bar{S}}\vec x_{\bar{S}}}\bra{O_{\bar{S},S}\left(\vec x_{S}+\vec z_{S}\right)+O_{\bar{S},\bar{S}}\vec x_{\bar{S}}}.
\end{equation}
Now, let us now take the square of $X=\left(r\left(O\right)^{\Gamma_{S}}\right)^{\dagger}r\left(O\right)^{\Gamma_{S}}$ in order to show that $X$ is proportional to a projector.
\begin{align}
    X^2 &=\sum_{\vec x_{S},\vec y_{S}}\: \sum_{\vec z_{S},\vec z'_{S}\in\ker O_{S,S}}\ket{\vec x_{S}}\bra{\vec y_{S}+\vec z'_{S}} \, \braket{\vec x_{S}+\vec z_{S}|\vec y_{S}} \nonumber \\
&\quad \otimes \sum_{\vec x_{\bar{S}}}\sum_{\vec y_{\bar{S}}}\ket{O_{\bar{S},S}\vec x_{S}+O_{\bar{S},\bar{S}}\vec x_{\bar{S}}}\bra{O_{\bar{S},S}\left(\vec y_{S}+\vec z'_{S}\right)+O_{\bar{S},\bar{S}}\vec y_{\bar{S}}}  \\
 &\quad \times \braket{O_{\bar{S},S}\left(\vec x_{S}+\vec z_{S}\right)+O_{\bar{S},\bar{S}}\vec x_{\bar{S}}|O_{\bar{S},S}\vec y_{S}+O_{\bar{S},\bar{S}}\vec y_{\bar{S}}}. \nonumber
\end{align}
We get non-vanishing contributions only when $\vec y_{S}= \vec x_{S}+\vec z_{S}$ and $\vec x_{\bar{S}}+ \vec y_{\bar{S}} \in \ker O_{\bar{S},\bar{S}}$ are satisfied and so we find:
\begin{align}
     X^2  &=\sum_{\vec x_{S},\vec y_{S}}\: \sum_{\vec z_{S},\vec z'_{S}\in\ker O_{S,S}}\ket{\vec x_{S}}\bra{\vec x_{S}+ \vec z_{S}+\vec z'_{S}} \nonumber \\
&\quad \otimes \sum_{\vec x_{\bar{S}}}\sum_{\vec r_{\bar{S}}\in \ker O_{\bar{S},\bar{S}}}\ket{O_{\bar{S},S}\vec x_{S}+O_{\bar{S},\bar{S}}\vec x_{\bar{S}}}\bra{O_{\bar{S},S}\left(\vec 
x_{S}+\vec z_{S}+\vec z'_{S}\right)+O_{\bar{S},\bar{S}}\left(\vec x_{\bar{S}}+\vec r_{\bar{S}}\right)} \\
&=\left|\ker O_{S,S}\right|\left|\ker O_{\bar{S},\bar{S}}\right| X \\
&= 2^{2\dim\left(\ker O_{S,S}\right)} X.
\end{align}
Here, in the last step, we used \Cref{lem:kernels-of-principal-submatrices}.
Hence $X$ is proportional to a projection and its eigenvalues lie in $\left\{ 0,2^{2\dim\left(\ker O_{S,S}\right)}\right\}$ and so the singular values of $r(O)^{\Gamma_S}$ lie in $\left\{ 0,2^{\dim\left(\ker O_{S,S}\right)}\right\}$. Lastly, to determine how many non-zero singular values there are, we compute the rank of $X$, i.e., we compute $\tr(X)$:
\begin{align}
    \tr(X) 
    &= \sum_{\vec x_{S},\vec x_{\bar{S}}}\sum_{\vec z_{S}\in \ker O_{S,S}} \tr\left( \ket{\vec x_{S}}\bra{\vec x_{S}+\vec z_{S}}  \otimes\ket{O_{\bar{S},S}\vec x_{S}+O_{\bar{S},\bar{S}}\vec x_{\bar{S}}}\bra{O_{\bar{S},S}\left(\vec x_{S}+\vec z_{S}\right)+O_{\bar{S},\bar{S}}\vec x_{\bar{S}}}\right) \\
     &= \sum_{\vec x_{S},\vec x_{\bar{S}}}\trace\left(\ket{\vec x_{S}}\bra{\vec x_{S}}\right)\trace\left(\ket{O_{\bar{S},S}\vec x_{S}+O_{\bar{S},\bar{S}}\vec x_{\bar{S}}}\bra{O_{\bar{S},S}\vec x_{S}+O_{\bar{S},\bar{S}}\vec x_{\bar{S}}}\right)\\
    &= 2^t\,.
\end{align}
On the other hand, $\tr(X)$ equals the sum of the non-zero eigenvalues of $X$. This implies that $X$ has $2^{t-2\dim\left(\ker O_{S,S}\right)}$ non-zero eigenvalues and so $r(O)^{\Gamma_S}$ has $2^{t-2\dim\left(\ker O_{S,S}\right)}$ non-zero singular values.
\end{proof}
\Cref{lem:singular-values-partial-transpose} provides a characterization of the singular values of $R(O)^{\Gamma_S}$ and hence the trace-norm $\lVert R\left(O\right)^{\Gamma_{S}}\rVert_{1} $  in terms of $\dim\left(\ker O_{S,S}\right)$. For our purposes, we are interested in proving upper bounds on $\lVert R\left(O\right)^{\Gamma_{S}}\rVert_{1}$ for all $O\in O_t$ except the identity $O=e$. Hence, we will now show that for all $O\neq e$ one can always choose the set $S$ on which the partial transpose is applied, such that $\dim\left(\ker O_{S,S}\right)$ is non-zero, i.e., $\ker O_{S,S}$ is non-trivial. To prove this, it is sufficient to show that $O$ has at least a single zero principal minor, i.e., there exists $S\subset [t]$ such that $\det\left(O_{S,S}\right)=0$. The following proposition implies this sufficient condition as a corollary:
\begin{proposition}[All principal minors equal 1 implies conjugate to upper triangular]
  For any binary matrix $A\in\mathbb{F}_2^{t\times t}$ such that $\det\left(A_{S,S}\right)=1$ for all $S\in\left[t\right]$, there exists $\pi\in\mathcal{S}_{t}$ such that $\pi A\pi^{-1}$ is upper triangular.
\end{proposition}
\begin{proof}
 We proceed by induction on $t$. For $t=1,2$, the claim can be easily verified. Now, consider a $t \times t$-matrix $A_t\in\mathbb{F}_2^{t\times t}$ such that $\det\left(A_{S,S}\right)=1$ for all $S\in\left[t\right]$. Since the $|S|=1$-minors correspond to the diagonal elements of $A_t$, all diagonal elements $a_{ii}=1$ for $i\in [t]$. We now write $A_t$ in block form singling out the last row and column as follows:
    \begin{equation}
        A_{t}=\left(\begin{array}{c|c}
A_{t-1} & \vec w\\
\hline \vec v^{T} & 1
\end{array}\right) \,.
    \end{equation}
    By the induction hypothesis, there exists some $\pi_1 \in \mathcal{S}_{t-1}\subseteq \mathcal{S}_t$ such that $\pi_1 A_{t-1} \pi_1^{-1}$ is upper triangular. Applying this permutation to $A_t$, we obtain
    \begin{equation*}
        \pi_1 A_{t}\pi_1^{-1}=\left(\begin{array}{c|c}
        \pi_1 & \vec 0\\
        \hline \vec 0 & 1
        \end{array}\right)\left(\begin{array}{c|c}
        A_{t-1} & \vec w\\
        \hline \vec v^{T} & 1
        \end{array}\right)\left(\begin{array}{c|c}
        \pi_1^{-1} & \vec 0\\
        \hline \vec 0 & 1
        \end{array}\right)=\left(\begin{array}{c|c}
        U_{t-1} & \pi_1 \vec w\\
        \hline \left(\pi_1 \vec v\right)^{T} & 1
        \end{array}\right)=:\left(\begin{array}{c|c}
        U_{t-1} & \tilde{\vec w}\\
        \hline \tilde{\vec v}^{T} & 1
        \end{array}\right) \,,
    \end{equation*}
where $U_{t-1}=\pi_1 A_{t-1} \pi_1^{-1}$ is a $(t-1)\times(t-1)$ upper triangular matrix. 

    Now, after applying this permutation, the matrix is almost upper triangular, except for the last row which still has $\tilde{\vec v}^{T}$ in it. Our goal is to remove any non-zero elements of $\tilde{\vec v}^T$ from that row by conjugation with additional permutations. In particular, we will now show that there exists a $\pi_2\in \mathcal{S}_t$ such that $\pi_2\pi_1 A_t \left(\pi_2\pi_1 \right)^{-1}$ is fully upper triangular.

    We start by examining the principal submatrix of $\pi_1 A_{t} \pi_1^{-1}$ corresponding to $S=\{t-1,t\}$. The corresponding principal minor is given by
    \begin{equation}
        \left|\begin{array}{cc}
1 & \tilde{w}_{t-1}\\
\tilde{v}_{t-1} & 1
\end{array}\right|=1+\tilde{w}_{t-1}\tilde{v}_{t-1}\,.
    \end{equation}
    By assumption, this principal minor is equal to 1, and so at least one of $\tilde{v}_{t-1}$ and $\tilde{w}_{t-1}$ must be zero.
    Hence, we now treat these two cases, separately:
\begin{enumerate}
    \item If $\tilde{v}_{t-1}=0$, then $ \pi_1 A_{t}\pi_1^{-1}$ is of the form 
            \begin{equation}
                 \pi_1 A_{t}\pi_1^{-1} = \left(\begin{array}{c|cc}
        U_{t-2} & \vec{u}_{[t-2]} & \tilde{\vec w}_{[t-2]}\\
        \hline \vec{0} & 1 & \tilde{w}_{t-1}\\
        \tilde{\vec v}_{[t-2]} & 0 & 1
        \end{array}\right).
            \end{equation}
            Now, consider the $(t-1)\times (t-1)$-principal submatrix of this matrix obtained by deleting the $(t-1)$-th row and column, namely
            \begin{equation}
                \left(\pi_1 A_{t}\pi_1^{-1}\right)_{S,S}=\left(\begin{array}{cc}
        U_{t-2} & \tilde{\vec w}_{[t-2]}\\
        \tilde{\vec v}_{[t-2]} & 1
        \end{array}\right),
            \end{equation}
            for $S=[t]\setminus\{t-1\}$. Since the principal minors of a matrix are invariant under conjugation by permutations, all principal minors of this matrix are equal to 1 by assumption. Hence, we can apply the induction hypothesis to this matrix to obtain a $\pi_2$ that will put this matrix into upper triangular form.
    \item If $\tilde{v}_{t-1}=1$, then $\tilde{w}_{t-1}=0$. In this case, we first apply the transposition $\sigma=(n-1,n)$ to get
    \begin{equation}
        \sigma\pi_1 A_{t}\left(\sigma\pi_1\right)^{-1}=\left(\begin{array}{c|cc}
U_{t-2} & \tilde{\vec w}_{[t-2]} & \vec{u}_{[t-2]}\\
\hline \tilde{\vec v}_{[t-2]} & 1 & 1\\
0 & 0 & 1
\end{array}\right).
    \end{equation}
    Now, consider the $(t-1)\times (t-1)$-principal submatrix of this matrix obtained by deleting the last row and column, namely
     \begin{equation}
        \left(\sigma\pi_1 A_{t}\left(\sigma\pi_1\right)^{-1}\right)_{S,S}
        =\left(\begin{array}{cc}
U_{t-2} & \tilde{\vec w}_{[t-2]} \\
 \tilde{\vec v}_{[t-2]} & 1 
\end{array}\right),
    \end{equation}
    for $S=[t-1]$. Then, we can apply the induction hypothesis to this matrix, to obtain a $\kappa\in\mathcal{S}_{t-1}$. Overall, we set $\pi_2 = \kappa \sigma$.
\end{enumerate}
    In both cases, the resulting matrix $\pi_2\pi_1 A_t \left(\pi_2\pi_1 \right)^{-1}$ is upper triangular as required.
\end{proof}
Note that for orthogonal matrices $O$ such that $O^T O = I$, $O$ being upper triangular implies that $O$ is the identity, i.e., $O=e$. Hence, we obtain the following corollary:
\begin{corollary}\label{cor:trace-norm-bound-partial-transpose}
For all $O\in O_{t}\setminus\left\{ e\right\}$, there exists $S\in\left[t\right]$ such that $\| R\left(O\right)^{\Gamma_{S}}\| _{1} \leq 2^{n(t-1)} $.
\end{corollary}

\subsection{On lower bounding via the tree representation framework}\label{ssec:tree-representation-framework}

In this section, we discuss a potential alternative route to showing a single-copy lower bound for stabilizer testing, namely the tree representation framework. While exploring this route, we found to our surprise, that there exist \emph{post-selective} single-copy algorithms for distinguishing stabilizer states from the maximally mixed state using only $O(1)$ many copies. This is in stark contrast to the case of purity testing considered in \cite{chenExponentialSeparationsLearning2022a}, as we explain in more detail now.\\

A range of single-copy lower bounds have been obtained by means of the \textit{tree representation framework} (see e.g. \cite{chenExponentialSeparationsLearning2022a, chenOptimalTradeoffsEstimating2024a} for detailed expositions). In summary, in this framework, the quantum learning/testing algorithm is modelled as a tree where nodes correspond to measurements and edges correspond to measurement outcomes.
By a post-selective single-copy algorithm, we mean a single-copy algorithm that can post-select on obtaining certain measurement outcomes and we discount any copies used that resulted in undesired measurement outcomes. In the tree representation framework this corresponds to algorithms that can deterministically follow a desired path (as long as it has nonzero probability) from the root to a leaf rather than probabilistically walking down the tree.\\

Particularly relevant for us is the $t=\Omega(2^{n/2})$ sample complexity lower bound on purity testing \cite{chenExponentialSeparationsLearning2022a, chenOptimalTradeoffsEstimating2024a} proved via this framework because of the similarities between purity testing and stabilizer testing mentioned throughout the manuscript. Here, purity testing refers to distinguishing between pure states and maximally mixed states, given copies of an unknown state that is promised to be one of the two. The key technical lemma behind this bound is the following:
\begin{lemma}[Lemma 5.12. in \cite{chenExponentialSeparationsLearning2022a}]
    Let $\Pi_{\mathrm{sym}}$ be the projector on the symmetric subspace, $\mathrm{Sym}_t((\mathbb{C}^2)^{\otimes n}) = \mathrm{span}\{\ket{\psi}^{\otimes t} :  \ket{\psi} \in (\mathbb{C}^2)^{\otimes n}\}$. Then, for any collection of pure states $\ket{\psi_1},\dots, \ket{\psi_t}\in (\mathbb{C}^2)^{\otimes n}$, it holds that
    \begin{equation}
        \tr \left(\Pi_{\mathrm{sym}} \bigotimes_{i=1}^t \ket{\psi_i}\bra{\psi_i} \right) \geq 1/t! \,.
    \end{equation}
    \label{lem:purity-testing-lemma}
\end{lemma}
To understand the significance of this bound, suppose that contrary to the lemma, there existed a collection of states such that $\tr \left(\Pi_{\mathrm{sym}} \bigotimes_{i=1}^t \ket{\psi_i}\bra{\psi_i} \right) = 0$. Then, no learning algorithm could produce these states as outcomes from measuring $\ket{\psi}^{\otimes t}$ when the unknown state is pure. In other words, these states (as measurement outcomes) are inconsistent with arising from $\ket{\psi}^{\otimes t}$ but fully consistent with measurements on the maximally mixed state $\mathds{1}^{\otimes t}/2^{nt}$. This would in turn imply the existence of a post-selective algorithm for purity testing which simply post-selects on obtaining this collection of states as outcomes.
From this chain of thought, we conclude that \Cref{lem:purity-testing-lemma} actually implies that the $t=\Omega(2^{n/2})$ single-copy lower bound for purity testing holds, even when considering post-selective single-copy algorithms.\\

The key finding of this section is that, unlike purity testing, post-selective algorithms for distinguishing stabilizer states from the maximally mixed state exist using just $O(1)$ copies. As a consequence, proving lower bounds on stabilizer testing via the tree representation framework seems daunting. Specifically, we provide collections of states with zero overlap with any $\ket{S}^{\otimes t}$, meaning these states are inconsistent with single-copy measurements on any stabilizer state. To achieve this, instead of the symmetric subspace, we focus on the subspace spanned by the $t$-th order tensor powers of stabilizer states,
\begin{equation}
    \mathrm{STAB}_{n,t} = \mathrm{span}\{\ket{S}^{\otimes t} \, : \, \ket{S} \in \stabn \}.
\end{equation}
This space has featured already in Section 5.2 of \cite{grossSchurWeylDualityClifford2021a} where it was shown that 
\begin{equation}
    \Pi_{O} := \frac{1}{|O_t|} \sum_{O\in O_t} R(O),
\end{equation}
is the orthogonal projector onto $\mathrm{STAB}_{n,t}$. For general $t$, this space is not well-understood, in particular, we do not know a basis for it in the regime $t\leq n$. However for  $t\leq 5$, we understand it, since in this case, $\Pi_{O} = \Pi_{\mathrm{sym}}$ simply because for $t\leq 5$, $O_t = S_t$. We now show that our understanding breaks down immediately after this point. \\

Consider the case of $n=1$ and $t=6$: we have that 
\begin{equation}
    \stab (1)=\left\{ \ket 0,\ket 1,\ket +,\ket -,\ket i,\ket{-i}\right\} \,.
\end{equation}
Now consider the state 
\begin{equation}
    \ket{\phi}=\ket 0\ket 1\ket +\ket -\ket i\ket{-i} \,.
\end{equation}
It is easy to see (but remarkable) that $\ket{\phi}$ is orthogonal to $\mathrm{STAB}_{1,t}$ and so $\Pi_{O}\ket{\phi} = 0$, providing a counter-example to a stabilizer version of \Cref{lem:purity-testing-lemma}. In fact, this basic example can be used as a building block to construct examples for arbitrary $n$ as follows:
\begin{lemma}
\label{lem:orthogonal-example-states}
    Let $t=6$, then consider the collection of states
    \begin{equation}
    \label{eq:padded-example}
        \left\{ \ket{S_{1}},\dots,\ket{S_{6}}\right\} =\left\{ \ket 0,\ket 1,\ket +,\ket -,\ket i,\ket{-i}\right\} \otimes\ket{0^{n-1}} \,.
    \end{equation}
    Then,
    \begin{equation}
        \tr \left(\Pi_{O}\bigotimes_{i=1}^{6}\ket{S_{i}}\bra{S_{i}} \right) = 0 \,.
    \end{equation}
\end{lemma}

\begin{proof}
Let $\ket{\phi}=\ket 0\ket 1\ket +\ket -\ket i\ket{-i}$, then using that $R\left(O\right)=r\left(O\right)^{\otimes n}$, we have
\begin{align}
\tr \left(\Pi_{O}\bigotimes_{i=1}^{6}\ket{S_{i}}\bra{S_{i}} \right) 
&=
\frac{1}{\left|O_{t}\right|}\sum_{O\in O_{t}}\trace\left(R\left(O\right)\bigotimes_{i=1}^{6}\ket{S_{i}}\bra{S_{i}}\right) \, ,\\ 
& =\frac{1}{\left|O_{t}\right|}\sum_{O\in O_{t}}\trace\left(r\left(O\right)\ket{\phi}\bra{\phi}\right) \cdot \trace\left(r\left(O\right)^{\otimes n-1}\left(\ket{0^{n-1}}\bra{0^{n-1}}\right)^{\otimes6}\right) \, ,\\
 & =\frac{1}{\left|O_{t}\right|}\sum_{O\in O_{t}}\trace\left(r\left(O\right)\ket{\phi}\bra{\phi}\right) \,, \\
 &= 0 \,.
\end{align}
where we used that $r\left(O\right)^{\otimes n-1}\ket{0^{n-1}}^{\otimes6}=\ket{0^{n-1}}^{\otimes6}$
for all $O\in O_{t}$ (recall from \Cref{fact:stabilizer-powers-sandwich} that $R\left(O\right)\ket S^{\otimes t}=\ket S^{\otimes t }$
for all $t$ and stabilizer states $\ket S\in\mathrm{Stab}\left(n\right)$).
\end{proof}

We can further generalize this example by noting for all $C\in \Cln$, and any collection of states $\ket{\psi_1},\dots, \ket{\psi_t}\in (\mathbb{C}^2)^{\otimes n}$, it holds that
\begin{equation}
\sum_{O\in O_{t}}\trace\left(R\left(O\right)\bigotimes_{i=1}^{t}\ket{\psi_{i}}\bra{\psi_{i}}\,\right)=\sum_{O\in O_{t}}\trace\left(R\left(O\right)\bigotimes_{i=1}^{t}C\ket{\psi_{i}}\bra{\psi_{i}}C^{\dagger}\,\right) \,.
\end{equation}
This follows immediately from the fact that 
$\left[R\left(O\right),C^{\otimes t}\right]=0$ for all $C \in \Cln$ and all $O \in O_t$. Hence, the example extends to an entire Clifford orbit of the collection given in \Cref{eq:padded-example}. We note that there are examples of states (even at $t=6$) that are not of this form, and we do not understand the general structure of states orthogonal to $\mathrm{Stab}(n,t)$. We leave this as a topic of future research.

\section*{Acknowledgements}
We thank Daniel Liang, Matthis Lehmkühler, Jonas Haferkamp, Felipe Montealegre, and Robert Huang for helpful discussions. MH acknowledges funding by the BMBF (MUNIQC-Atoms). JH acknowledges funding from the Dutch Research Council (NWO) through Veni No.VI.Veni.222.331 and the Quantum Software Consortium (NWO Zwaartekracht Grant No.024.003.037).

\bibliographystyle{alphaurl}
\bibliography{refs}

\appendix

\section{Appendix}\label{appendix:proof-of-relation-r-p}

\subsection{Proof of Lemma \ref{lem:relation-r-p}\label{subsec:Proof-of-relation-lemma}}

In this section, we will prove \Cref{lem:relation-r-p} which
we restate here for convenience:
\begin{lemma}[Subspace weight correspondence between $r_{\psi}$ and $p_{\psi}$]
\label{lem:relation-r-p-restated}Let $\ket \psi$ be a pure $n$-qubit
quantum state. Given a subspace $H\subseteq\mathbb{F}_{2}^{n}$, consider
its orthogonal complement  $H^{\perp}$ (with respect to the standard
inner product on $\mathbb{F}_{2}^{n}$). Then,
\begin{equation}
r_{\psi}\left(H\right)=\left|H\right|\,p_{\psi}\left(\vec 0_{n}\times H^{\perp}\right),
\end{equation}
where $\vec 0_n\times H^{\perp}\subseteq\mathcal{Z}$.
\end{lemma}
Our starting point for proving this lemma is the following relation
between $r_{\psi}$ and $p_{\psi}$ proved in \cite{grewalEfficientLearningQuantum2024d}:
\begin{lemma}[Proposition 8.4. in \cite{grewalEfficientLearningQuantum2024d}]
\label{lem:correspondence-r-p-individual}
Let $\ket \psi$ be a pure
$n$-qubit quantum state and let $\vec a\in\mathbb{F}_{2}^{n}$. Then 
\begin{equation}
r_{\psi}\left(\vec a\right)=\sum_{\vec b\in\mathbb{F}_{2}^{n}}p_{\psi}\left(\vec a,\vec b\right)\,.
\end{equation}
\end{lemma}
Furthermore, we will need two more additional ingredients. First,
we will need the following fact:
\begin{fact}
\label{fact:projection_on_complement}
For any subspace $H\subseteq\mathbb{F}_{2}^{n}$
and a fixed $\vec v\in\mathbb{F}_{2}^{n}$,
\begin{equation}
\sum_{\vec a\in H}\left(-1\right)^{\vec a\cdot \vec v}=\begin{cases}
\left|H\right| & \vec v\in H^{\perp}\\
0 & \mathrm{else}
\end{cases}\,.
\end{equation}
\end{fact}

Secondly, we need the fact that the Fourier transform of the characteristic
distribution $p_{\psi}$ with respect to the standard inner product
is given by $\hat{p}_{\psi}\left(\vec v,\vec w\right)=\frac{1}{2^{n}}p_{\psi}\left(\vec w,\vec v\right)$
as proved in Proposition 17 in \cite{grewalLowStabilizerComplexityQuantumStates2023a}.
\begin{lemma}[Proposition 17 in \cite{grewalLowStabilizerComplexityQuantumStates2023a}]
Let $\ket \psi$ be an $n$-qubit pure state and let $\hat{p}_{\psi}\left(\vec v, \vec w\right)$
denote the Fourier transform of $p_{\psi}\left(\vec a,\vec b\right)$ with respect
to the standard inner product on $\mathbb{F}_{2}^{n}$, such that
\begin{equation}
p_{\psi}\left(\vec a, \vec b\right)=\sum_{\vec v, \vec w\in\mathbb{F}_{2}^{n}}\hat{p}_{\psi}\left(\vec v,\vec w\right)\left(-1\right)^{\left(\vec a, \vec b\right)\cdot\left(\vec v, \vec w\right)}\,.
\end{equation}
Then,
\begin{equation}
\hat{p}_{\psi}\left(\vec v,\vec w\right)=\frac{1}{2^{n}}p_\psi\left(\vec w, \vec v\right).
\end{equation}
\end{lemma}

With these ingredients, we are ready to prove \Cref{lem:relation-r-p-restated}:

\begin{proof}[Proof of Lemma \ref{lem:relation-r-p-restated}]
 By \Cref{lem:correspondence-r-p-individual}, for the subspace
$H\subseteq\mathbb{F}_{2}^{n}$, we have
\begin{equation}
r_{\psi}\left(H\right)=\sum_{\vec a\in H}\sum_{\vec b\in\mathbb{F}_{2}^{n}}p_{\psi}\left(\vec a, \vec b\right)\,,
\end{equation}
and plugging in the Fourier transform $p_{\psi}\left(\vec a, \vec b\right)=\sum_{\vec v,\vec w}\hat{p}_\psi \left(\vec v, \vec w\right)\left(-1\right)^{\vec a\cdot \vec v+ \vec b\cdot \vec w}$
we find
\begin{equation}
r_{\psi}\left(H\right)=\sum_{\vec v\in\mathbb{F}_{2}^{n}}\sum_{\vec w\in\mathbb{F}_{2}^{n}}\hat{p}_\psi\left(\vec v,\vec w\right)\sum_{\vec a\in H}\left(-1\right)^{\vec a\cdot v}\sum_{\vec b\in\mathbb{F}_{2}^{n}}\left(-1\right)^{\vec b\cdot \vec w}.
\end{equation}
Now, by  \Cref{fact:projection_on_complement}, we have that
\begin{equation}
\sum_{\vec b\in\mathbb{F}_{2}^{n}}\left(-1\right)^{\vec b\cdot \vec w}=\begin{cases}
2^{n} & \vec w = \vec 0_n,\\
0 & \mathrm{else}.
\end{cases}
\end{equation}
and furthermore,
\begin{equation}
\sum_{\vec a\in H}\left(-1\right)^{\vec a\cdot \vec v}=\begin{cases}
\left|H\right| & v\in H^{\perp},\\
0 & \mathrm{else}.
\end{cases}\,
\end{equation}
Hence, using $\hat{p}_\psi \left(\vec v, \vec w\right)=\frac{1}{2^{n}}p_\psi \left(\vec w, \vec v\right)$,
we find that
\begin{align}
r_{\psi}\left(H\right) & =\left|H\right|\sum_{v\in v\in H^{\perp}}p_\psi\left(\vec 0_{n},v\right)\\
 & =\left|H\right|\,p_{\psi}\left(\vec 0_{n}\times H^{\perp}\right)\,.
\end{align}
as claimed.
\end{proof}

\end{document}